\documentclass{llncs}

\usepackage{multicol}

\usepackage{amsmath}
\usepackage{amsfonts}
\usepackage{amssymb}
\usepackage{graphicx}
\usepackage{epic}
\usepackage{eepic}
\usepackage{epsfig,float}
\usepackage{pdfsync}
\renewcommand{\le}{\leqslant}
\renewcommand{\ge}{\geqslant}

\newcommand{\eps}{\varepsilon}

\newcommand{\Sig}{\Sigma}
\newcommand{\sig}{\sigma}
\newcommand{\noin}{\noindent}
\newcommand{\bi}{\begin{itemize}}
\newcommand{\ei}{\end{itemize}}
\newcommand{\be}{\begin{enumerate}}
\newcommand{\ee}{\end{enumerate}}
\newcommand{\bd}{\begin{description}}
\newcommand{\ed}{\end{description}}
\newcommand{\bq}{\begin{quote}}
\newcommand{\eq}{\end{quote}}
\newcommand{\txt}[1]{\mbox{ #1 }}
\newcommand{\defeq}{\stackrel{\rm def}{=}}

\newcommand{\ie}{\mbox{\it i.e.}}

\newcommand{\trng}{\mathrm{rng}}
\newcommand{\tdom}{\mathrm{dom}}
\newcommand{\trank}{\mathrm{rank}}

\newcommand{\tid}{\mbox{{\bf 1}}}

\newcommand{\cA}{{\mathcal A}}
\newcommand{\cB}{{\mathcal B}}
\newcommand{\cC}{{\mathcal C}}

\newcommand{\cE}{{\mathcal E}}
\newcommand{\cF}{{\mathcal F}}
\newcommand{\cG}{{\mathcal G}}

\newcommand{\cP}{{\mathcal P}}
\newcommand{\cM}{{\mathcal M}}
\newcommand{\cN}{{\mathcal N}}
\newcommand{\cT}{{\mathcal T}}
\newcommand{\cS}{{\mathcal S}}
\newcommand{\gL}{{\mathcal L}}
\newcommand{\gR}{{\mathcal R}}
\newcommand{\gJ}{{\mathcal J}}
\newcommand{\gH}{{\mathcal H}}

\newcommand{\rmPi}{{\mathrm \Pi}} 
\newcommand{\Lra}{{\Leftrightarrow}}
\newcommand{\lra}{{\leftrightarrow}}
\newcommand{\ra}{{\rightarrow}}

\newcommand{\lraL}{{\mathbin{\approx_L}}}

\newcommand{\meet}{{\mathbin{\wedge}}}
\newcommand{\join}{{\mathbin{\vee}}}

\newcommand{\qedb}{\hfill$\blacksquare$}

\newcommand{\rev}{\mathbb{R}}
\newcommand{\deter}{\mathbb{D}}

\DeclareMathOperator{\Fix}{Fix}
\DeclareMathOperator{\Orbit}{\Omega}
\DeclareMathOperator{\Max}{Max} 

\title{Syntactic Complexity of $\gR$- and $\gJ$-Trivial Regular Languages\thanks{This work was supported by the Natural Sciences and Engineering Research Council of Canada under grant No.~OGP0000871 and  a Postgraduate Scholarship.
}
}

\author{Janusz~Brzozowski, Baiyu Li\thanks{Present address: 
Optumsoft, Inc.,
275 Middlefield Rd, Suite 210, Menlo Park, CA 94025, USA}
}
\authorrunning{Brzozowski, Li}   

\institute{David R. Cheriton School of Computer Science, University of Waterloo \\
Waterloo, ON, Canada N2L 3G1\\
\{{\tt \{brzozo, b5li\}@uwaterloo.ca} \}
}

\begin{document}

\maketitle
\begin{abstract}
The syntactic complexity of a subclass of the class of regular languages is the maximal cardinality of syntactic semigroups of languages in that class, taken as a function of the state complexity $n$ of these languages.
We prove that $n!$ and $\lfloor e(n-1)! \rfloor$ are tight upper bounds for 
the syntactic complexity of $\gR$- and $\gJ$-trivial regular languages, respectively. 
We also prove that $2^{n-1}$ is the tight upper bound on the state complexity of reversal of $\gJ$-trivial regular languages. 
\bigskip

\noin
\end{abstract}

\section{Introduction}\label{sec:intro}

The \emph{state complexity} of a regular language $L$ is the number of states in the minimal deterministic finite automaton (DFA) accepting $L$. An equivalent notion is \emph{quotient complexity}, which is the number of distinct left quotients of $L$. 
The \emph{syntactic complexity} of $L$ is the cardinality of the syntactic semigroup of $L$. 
Since the syntactic semigroup of $L$ is isomorphic to the semigroup of transformations performed by the minimal DFA of $L$, it is natural to consider the relation between syntactic complexity and state complexity. The \emph{syntactic complexity of a subclass of regular languages} is the maximal syntactic complexity of languages in that class, taken as a function of the state complexity of these languages.

Here we consider the classes of languages defined using the well-known Green equivalence relations on semigroups~\cite{Pin97}. Let $M$ be a monoid, that is, a semigroup with an identity, and let $s, t \in M$ be any two elements of $M$. The Green relations on $M$, denoted by $\gL,\gR,\gJ$ and $\gH$, are defined as follows: 
\begin{align*}
 s \mathbin{\gL} t &\mathbin{\Lra} Ms = Mt, \\
 s \mathbin{\gR} t &\mathbin{\Lra} sM = tM, \\
 s \mathbin{\gJ} t &\mathbin{\Lra} MsM = MtM, \\
 s \mathbin{\gH} t &\mathbin{\Lra} s \mathbin{\gL} t \txt{and} s \mathbin{\gR} t.
\end{align*}
If $\rho \in \{\gL, \gR, \gJ, \gH\}$ is an equivalence relation on $M$, then $M$ is \emph{$\rho$-trivial} if and only if $(s, t) \in \rho$ implies $s = t$ for all $s, t \in M$. A language is \emph{$\rho$-trivial} if and only if its syntactic monoid is $\rho$-trivial. 
In this paper we consider only regular $\rho$-trivial languages. 
$\gH$-trivial regular languages are exactly the star-free languages~\cite{Pin97}, and $\gL$-, $\gR$-, and $\gJ$-trivial regular languages are all subclasses of star-free languages. 
The class of $\gJ$-trivial languages is the intersection of $\gR$- and $\gL$-trivial classes. 

A language $L \subseteq \Sig^*$ is \emph{piecewise-testable} if it is a finite boolean combination of languages of the form $\Sig^* a_1 \Sig^* \cdots \Sig^* a_l \Sig^*$, where $a_i \in \Sig$. Simon~\cite{Sim72,Sim75} proved in 1972 that a language is piecewise-testable if and only if it is $\gJ$-trivial. A \emph{biautomaton} is a finite automaton which can read the input word alternatively from left and  right. In 2011 Kl\'{\i}ma and Pol{\'a}k~\cite{KP11} showed that a language is piecewise-testable if and only if it is accepted by an acyclic biautomaton; here self-loops are allowed, as they are not considered cycles. 

In 1979 Brzozowski and Fich~\cite{BrFi80} proved that a regular language is $\gR$-trivial if and only if its minimal DFA is \emph{partially ordered}, that is, it is acyclic as above. They also showed that $\gR$-trivial regular languages are  finite boolean combinations of languages $\Sig_1^* a_1\Sig^* \cdots \Sig_l^* a_l \Sig^*$, where $a_i \in \Sig$ and $\Sig_i \subseteq \Sig \setminus \{a_i\}$. Recently Jir{\'a}skov{\'a} and Masopust proved a tight upper bound on the state complexity of reversal of $\gR$-trivial languages~\cite{JiMa12}. 

The syntactic complexity of the following subclasses of regular languages was considered: In 1970 Maslov~\cite{Mas70} noted that $n^n$ was a tight upper bound on the number of transformations performed by a DFA of $n$ states. In 2003--2004, Holzer and K\"onig~\cite{HoKo04}, and Krawetz, Lawrence and Shallit~\cite{KLS03} studied unary and binary languages. In 2010 Brzozowski and Ye~\cite{BrYe11} examined ideal and closed regular languages. In 2012 Brzozowski, Li and Ye studied prefix-, suffix-, bifix-, and factor-free regular languages~\cite{BLY12}, Brzozowski and Li~\cite{BL12a} considered the class of star-free languages and three of its subclasses, and Brzozowski and Liu~\cite{BrLiu12} studied finite/cofinite, definite, and reverse definite languages, where 
$L$ is \emph{definite} (\emph{reverse-definite}) if it can be decided whether a word $w$ belongs to $L$ by examining the suffix (prefix) of $w$ of some fixed length.

We state basic definitions and facts in Section~\ref{sec:pre}. In Sections~\ref{sec:Rtrivial} and~\ref{sec:Jtrivial} we prove tight upper bounds on the syntactic complexities of $\gR$- and $\gJ$-trivial regular languages, respectively. 
In Section~\ref{sec:rev} we prove the tight upper bound on the quotient complexity of reversal of $\gJ$-trivial regular languages, and we show that this bound can be met by our languages with maximal syntactic complexities. 
Section~\ref{sec:con} concludes the paper. 
%

\section{Preliminaries}\label{sec:pre}

Let $Q$ be a non-empty finite set with $n$ elements, and assume without loss of generality that $Q = \{1,2,\ldots, n\}$. There is a linear order on $Q$, namely the natural order $<$ on integers. If $X$ is a non-empty subset of $Q$, then the maximal element in $X$ is denoted by $\max(X)$. A \emph{partition} $\pi$ of $Q$ is a collection $\pi = \{X_1, X_2, \ldots, X_m\}$ of non-empty subsets of $Q$ such that  $Q = X_1 \cup X_2 \cup \cdots \cup X_m$, and
$X_i \cap X_j = \emptyset$ for all $1 \le i < j \le m$.
%
%
We call each subset $X_i$ a \emph{block} in $\pi$. For any partition $\pi$ of $Q$, let $\Max(\pi) = \{\max(X) \mid X \in \pi\}$. The set of all partitions of $Q$ is denoted by $\rmPi_Q$. We define a partial order $\preceq$ on $\rmPi_Q$ such that, for any $\pi_1, \pi_2 \in \rmPi_Q$, $\pi_1 \preceq \pi_2$ if and only if each block of $\pi_1$ is contained in some block of $\pi_2$. We say $\pi_1$ \emph{refines} $\pi_2$ if $\pi_1 \preceq \pi_2$. 
The poset $(\rmPi_Q, \preceq)$ is a finite lattice: For any $\pi_1, \pi_2 \in \rmPi_Q$, the \emph{meet} $\pi_1 \meet \pi_2$ is the $\preceq$-largest partition that refines both $\pi_1$ and $\pi_2$, and the \emph{join} $\pi_1 \join \pi_2$ is the $\preceq$-smallest partition that is refined by both $\pi_1$ and $\pi_2$. From now on, we refer to the lattice $(\rmPi_Q, \preceq)$ simply  as $\rmPi_Q$.

A {\em transformation} of a set $Q$ is a mapping of $Q$ into itself. We consider only transformations $t$ of a finite set $Q$. If  $i \in Q$, then $it$ is the {\it image} of $i$ under $t$.  If $X$ is a subset of $Q$, then $Xt = \{it \mid i \in X\}$, and the {\em restriction} of $t$ to $X$, denoted by $t|_X$, is a mapping from $X$ to $Xt$ such that $it|_X = it$ for all $i \in X$. The {\em composition} of  transformations $t_1$ and $t_2$ of $Q$ is a transformation $t_1 \circ t_2$ such that $i (t_1 \circ t_2) = (i t_1) t_2$ for all $i \in Q$. We~usually drop the  operator ``$\circ$'' and write $t_1t_2$ for short. 
An arbitrary transformation can be written in the form
\begin{equation*}\label{eq:transmatrix}
t=\left( \begin{array}{ccccc}
1 & 2 &   \cdots &  n-1 & n \\
i_1 & i_2 &   \cdots &  i_{n-1} & i_n
\end{array} \right ),
\end{equation*}
where $i_k = kt$,  $1\le k\le n$, and $i_k\in Q$. We also use the notation $t = [i_1,i_2,\ldots,i_n]$ for $t$ above. The {\em domain} $\tdom(t)$ of $t$ is $Q.$
The {\em range} $\trng(t)$ of $t$ is the set $\trng(t) = Q t.$ The \emph{rank} $\trank(t)$ of $t$ is the cardinality of $\trng(t)$, \ie, $\trank(t) = |\trng(t)|$. The binary relation $\omega_t$ on $Q \times Q$ is defined as follows: For any $i, j \in Q$, $i \mathbin{\omega_t} j$ if and only if $it^k = jt^l$ for some $k, l \ge 0$. This is an equivalence relation, and each equivalence class is called an \emph{orbit} of $t$. For any $i \in Q$, the orbit of $t$ containing $i$ is denoted by $\omega_t(i)$. The~set of all orbits of $t$ is denoted by $\Orbit(t)$. Clearly, $\Orbit(t)$ is a partition of $Q$. 

A \emph{permutation} of $Q$ is a mapping of $Q$ \emph{onto} itself, so here $\trng(\pi) = Q$. 
The \emph{identity} transformation $\tid_Q$ maps each element to itself. 
A transformation $t$ is a \emph{cycle} of length $k$, where $k \ge 2$, if there exist pairwise different elements $i_1,\ldots,i_k$ such that
$i_1t=i_2, i_2t=i_3,\ldots, i_{k-1}t=i_k$, and $i_kt=i_1$, and the remaining elements are mapped to themselves.
A~cycle is denoted by $(i_1,i_2,\ldots,i_k)$.
For $i<j$, a \emph{transposition} is the cycle $(i,j)$.
A~\emph{singular} transformation, denoted by $i\choose j$, has $it=j$ and $ht=h$ for all $h\neq i$.
A~\emph{constant} transformation,  denoted by $Q \choose j$, has $it=j$ for all $i$.
A transformation $t$  is an \emph{idempotent} if $t^2 = t$.
The set $\cT_Q$ of all transformations of $Q$ is a finite semigroup, in fact, a monoid. We refer the reader to the book of Ganyushkin and Mazorchuk~\cite{GaMa09} for a detailed discussion of finite transformation semigroups. 

\medskip

For background about regular languages, we refer the reader to~\cite{Yu97}. Let $\Sig$ be a non-empty finite alphabet. Then $\Sig^*$ is the free monoid generated by $\Sig$, and $\Sig^+$ is the free semigroup generated by $\Sig$. A \emph{word} is any element of $\Sig^*$, and the empty word is $\eps$. The length of a word $w\in \Sig^*$ is $|w|$. A \emph{language} over $\Sig$ is any subset of $\Sig^*$. The \emph{reverse of a word} $w$ is denoted by $w^R$. For a language $L$, its \emph{reverse} is $L^R = \{w \mid w^R \in L\}$. The \emph{left quotient}, or simply \emph{quotient}, of a language $L$ by a word $w$ is   $L_w=\{x\in \Sig^*\mid wx\in L \}$. 

The \emph{Myhill congruence}~\cite{Myh57} $\lraL$ of any language $L$ is defined as follows:
\begin{equation*}
x~\lraL~y \mbox{ if and only if } uxv\in L  \Leftrightarrow uyv\in L\mbox { for all } u,v\in\Sig^*.
\end{equation*}
This congruence is also known as the \emph{syntactic congruence} of $L$. The quotient set $\Sig^+/ \lraL$ of equivalence classes of the relation $\lraL$ is a semigroup called the \emph{syntactic semigroup} of $L$, and $\Sig^*/ \lraL$ is the \emph{syntactic monoid} of~$L$. 
The \emph{syntactic complexity} $\sig(L)$ of $L$ is the cardinality of its syntactic semigroup.
A language is regular if and only if its syntactic semigroup is finite. We consider only regular languages, and so assume that all syntactic semigroups and  monoids are finite.

A DFA is denoted by $\cA = (Q, \Sig, \delta, q_1, F)$, as usual. The DFA $\cA$ accepts a word $w \in \Sigma^*$ if ${\delta}(q_1,w)\in F$. The language accepted by $\cA$ is denoted by $L(\cA)$. If $q$ is a state of $\cA$, then the language $L_q$ of $q$ is the language accepted by the DFA $(Q,\Sigma,\delta,q,F)$. Two states $p$ and $q$ of $\cA$ are \emph{equivalent} if $L_p = L_q$. If $L \subseteq \Sig^*$ is a regular language, then its \emph{quotient DFA} is $\cA = (Q, \Sig, \delta, q_1, F)$, where $Q=\{L_w \mid w\in\Sig^*\}$, $\delta(L_w,a)=L_{wa}$, $q_1=L_\eps=L$,  $F=\{L_w \mid \eps \in L_w\}$. The \emph{quotient complexity} $\kappa(L)$ of $L$ is the number of distinct quotients of $L$. The quotient DFA of $L$ is the minimal DFA accepting $L$, and so quotient complexity is the same as state complexity.

If $\cA = (Q, \Sig, \delta, q_1, F)$ is a DFA, then its \emph{transition semigroup}~\cite{Pin97}, denoted by $T_{\cA}$, consists of all transformations $t_w$ on $Q$ performed by non-empty words $w \in \Sig^+$ such that $it_w = \delta(i, w)$ for all $i \in Q$. The syntactic semigroup $T_L$ of a regular language $L$ is isomorphic to the transition semigroup of the quotient DFA $\cA$ of $L$~\cite{McNP71}, and we represent elements of $T_L$ by transformations in $T_{\cA}$. 
Given a set $G = \{t_a \mid a \in \Sig\}$ of transformations of $Q$, we can define the transition function $\delta$ of some DFA $\cA$ such that $\delta(i, a) = it_a$ for all $i \in Q$. The transition semigroup of such a DFA is the semigroup generated by $G$. When the context is clear, we write $a = t$,  to mean that the transformation performed by $a \in \Sig$ is~$t$.

\section{$\gR$-Trivial Regular Languages}\label{sec:Rtrivial} 

Given DFA $\cA = (Q, \Sig, \delta, q_1, F)$, we  define the \emph{reachability relation} $\ra$ as follows. For all $p, q \in Q$, $p \mathbin{\ra} q$ if and only if $\delta(p, w) = q$ for some $w \in \Sig^*$. We~say that $\cA$ is \emph{partially ordered}~\cite{BrFi80} if the relation $\ra$ is a partial order on $Q$. 

Consider the natural order $<$ on $Q$. A transformation $t$ of $Q$ is \emph{non-decreasing} if $p \le pt$ for all $p \in Q$. The set $\cF_Q$ of all non-decreasing transformations of $Q$ is a semigroup, since the composition of two non-decreasing transformations is again non-decreasing. It was shown in~\cite{BrFi80} that a language $L$ is $\gR$-trivial if and only if its quotient DFA is partially ordered. Equivalently, $L$ is an $\gR$-trivial language if and only if its syntactic semigroup contains only non-decreasing transformations.

 It is known~\cite{GaMa09} that  $\cF_Q$ is generated by the following set 
$$\cG\cF_Q = \{\tid_Q\} \cup \{t \in \cF_Q \mid t^2 = t \txt{and} \trank(t) = n-1\}.$$
For any transformation $t$ of $Q$, let $\Fix(t) = \{i \in Q \mid it = i\}$. Then

\begin{lemma}\label{lem:fixrank} 
For any $t \in \cG\cF_Q$, $\trng(t) = \Fix(t)$. 
\end{lemma}

\begin{proof} 
Pick arbitrary $t \in \cG\cF_Q$. The claim holds trivially for $\tid_Q$. Assume $t \neq \tid_Q$. Clearly $\Fix(t) \subseteq \trng(t)$. Suppose there exists $i \in \trng(t)$ but $it \neq i$. Then $jt = i$ for some $j \in Q$, and $j \neq i$. However, since $jt^2 = it \neq i = jt$, $t$ is not an idempotent, which is a contradiction. Therefore $\trng(t) = \Fix(t)$. \qed
\end{proof}

If $n = 1$, then $\cF_Q$ contains only the identity transformation $\tid_Q$, and $\cG\cF_Q = \cF_Q = \{\tid_Q\}$. So $|\cG\cF_Q| = |\cF_Q| = 1$. If $n \ge 2$, then we have

\begin{lemma}\label{lem:carGF}
For $n \ge 2$, $|\cG\cF_Q| = 1 + C^n_2$. 
\end{lemma}

\begin{proof} 
Pick $t \in \cG\cF_Q$ such that $t \neq \tid_Q$. Then $\trank(t) = n-1$, and, by Lemma~\ref{lem:fixrank}, $|\Fix(t)| = n-1$. There is only one element $i \in Q \setminus \Fix(t)$, and $i < it$. Note that $t$ is fully determined by the pair $(i, it)$. Hence there are $C^n_2$ different $t$. Together with the identity $\tid_Q$, the cardinality of $\cG\cF_Q$ is $1 + C^n_2$. \qed
\end{proof}

\begin{lemma}\label{lem:minGF} 
If $G \subseteq \cF_Q$ and $G$ generates $\cF_Q$, then $\cG\cF_Q \subseteq G$. 
\end{lemma}

\begin{proof}
Suppose there exists $t \in \cG\cF_Q$ such that $t \not\in G$. Since $G$ generates $\cF_Q$, $t$ can be written as $t = g_1 \cdots g_k$ for some $g_1,\ldots,g_k \in G$, where $k \ge 2$. Then $\trng(g_1) \supseteq \trng(g_1g_2) \supseteq \cdots \supseteq \trng(g_1g_2 \cdots g_k) = \trng(t)$. Note that $\tid_Q$ is the only element in $\cF_Q$ with range $Q$; so if $t = \tid_Q$, then $g_1 = \cdots = g_k = \tid_Q$, a~contradiction. 

Assume $t \neq \tid_Q$, and $g_i \neq \tid_Q$ for all $i=1,\ldots,k$. Then $\trank(t) = n-1$, and $\trng(g_1) = \cdots = \trng(g_k) = \trng(t)$. Since each $g_i$ is non-decreasing, for all $p \in \Fix(t)$, we must have $p \in \Fix(g_i)$ as well; so $\Fix(t) \subseteq \Fix(g_i)$. Moreover, since $\Fix(g_i) \subseteq \trng(g_i) = \trng(t)$ and $\trng(t) = \Fix(t)$ by Lemma~\ref{lem:fixrank}, $\Fix(g_i) = \Fix(t) = \trng(t)$. Now, let $q$ be the unique element in $Q \setminus \Fix(t)$. Then $q \not\in \Fix(g_1)$, and $qg_1 \in \Fix(g_2) = \cdots = \Fix(g_k)$. So~$q(g_1 \cdots g_k) = qg_1$. However, since $t = g_1 \cdots g_k$, $q(g_1 \cdots g_k) = qt$ and $qg_1 = qt$. Hence $g_1 = t$, and we get a contradiction again. Therefore $\cG\cF_Q \subseteq G$. \qed
\end{proof}

Consequently, $\cG\cF_Q$ is the unique minimal generator of $\cF_Q$. So we obtain  

\begin{theorem}\label{thm:Rtrivial}
If $L \subseteq \Sig^*$ is an $\gR$-trivial regular language of quotient complexity $\kappa(L) = n \ge 1$, then its syntactic complexity $\sigma(L)$ satisfies $\sigma(L) \le n!$, and this bound is tight if $|\Sig| = 1$ for $n = 1$ and $|\Sig| \ge 1 + C^n_2$ for $n \ge 2$.  
\end{theorem}

\begin{proof}
Let $\cA$ be the quotient DFA of $L$, and let $T_L$ be its syntactic semigroup. Then $T_L$ is a subset of $\cF_Q$. Pick an arbitrary $t \in \cF_Q$. For each $p \in Q$, since $p \le pt$, $pt$ can be chosen from $\{p, p+1, \ldots, n\}$. Hence there are exactly $n!$ transformations in $\cF_Q$, and $\sigma(L) \le n!$. 

When $n = 1$, the only regular languages are $\eps$ or $\emptyset$, and they are both  $\gR$-trivial. To see the bound is tight for $n \ge 2$, let $\cA_n = (Q, \Sig, \delta, 1, \{n\})$ be the DFA with alphabet $\Sig$ of size $1+C^n_2$ and set of states $Q = \{1,\ldots,n\}$, where each $a \in \Sig$ defines a distinct transformation in $\cG\cF_Q$. For each $p \in Q$, let $t_p = [p,n,\ldots,n]$. Since $\cG\cF_Q$ generates $\cF_Q$ and $t_p \in \cF_Q$, $t_p = e_1 \cdots e_k$ for some $e_1,\ldots,e_k \in \cG\cF_Q$, where $k$ depends on $p$. Then there exist $a_1,\ldots,a_k \in \Sig$ such that each $a_i$ performs $e_i$ and state $p$ is reached by $w = a_1 \cdots a_k$. Moreover, $n$ is the only final state of~$\cA_n$. Consider any non-final state $q \in Q \setminus \{n\}$. Since $t = [2,3,\ldots,n,n] \in \cF_Q$, there exist $b_1, \ldots, b_l \in \Sig$ such that the word $u = b_1 \cdots b_l$ performs $t$. State $q$ can be distinguished from other non-final states by the word $u^{n-q}$. Hence $L = L(\cA_n)$ has quotient complexity $\kappa(L) = n$. The syntactic monoid of $L$ is $\cF_Q$, and so $\sigma(L) = n!$. By Lemma~\ref{lem:minGF}, the alphabet of $\cA_n$ is minimal. \qed
\end{proof}

\begin{example}\label{ex:Rtrivial}
When $n = 4$, there are $4! = 24$ non-decreasing transformations of $Q = \{1,2,3,4\}$. Among them, there are 11 transformations with rank $n-1 = 3$. The following 6 transformations from the 11 are idempotents: 
\begin{align*}
  e_1 = [1, 2, 4, 4], & \qquad e_2 = [1, 3, 3, 4], \\
  e_3 = [1, 4, 3, 4], & \qquad e_4 = [2, 2, 3, 4], \\
  e_5 = [3, 2, 3, 4], & \qquad e_6 = [4, 2, 3, 4].
\end{align*}
Together with the identity transformation $\tid_Q$, we have the generating set $\cG\cF_Q$ for $\cF_Q$ with 7 transformations. We can then define the DFA $\cA_4$ with 7 inputs as in the proof of Theorem~\ref{thm:Rtrivial}; $\cA_4$ is shown in Fig.~\ref{fig:RTDFA}. The quotient complexity of $L = L(\cA_4)$ is $4$, and the syntactic complexity of $L$ is $24$. \qedb
\end{example}

\begin{figure}[hbt]
\begin{center}
\setlength{\unitlength}{0.00065617in}
\begingroup\makeatletter\ifx\SetFigFont\undefined%
\gdef\SetFigFont#1#2#3#4#5{%
  \reset@font\fontsize{#1}{#2pt}%
  \fontfamily{#3}\fontseries{#4}\fontshape{#5}%
  \selectfont}%
\fi\endgroup%
{\renewcommand{\dashlinestretch}{30}
\begin{picture}(4070,2116)(0,-10)
\put(1587.000,1110.250){\arc{337.500}{2.4981}{6.9267}}
\blacken\path(1725.575,1132.641)(1722.000,1009.000)(1783.339,1116.413)(1725.575,1132.641)
\put(1587,829){\ellipse{450}{450}}
\put(2712.000,1110.250){\arc{337.500}{2.4981}{6.9267}}
\blacken\path(2850.575,1132.641)(2847.000,1009.000)(2908.339,1116.413)(2850.575,1132.641)
\put(2712,829){\ellipse{450}{450}}
\put(3837.000,1110.250){\arc{337.500}{2.4981}{6.9267}}
\blacken\path(3975.575,1132.641)(3972.000,1009.000)(4033.339,1116.413)(3975.575,1132.641)
\put(3837,829){\ellipse{450}{450}}
\put(462.000,1110.250){\arc{337.500}{2.4981}{6.9267}}
\blacken\path(600.575,1132.641)(597.000,1009.000)(658.339,1116.413)(600.575,1132.641)
\put(2149.500,232.750){\arc{3270.249}{3.6052}{5.8195}}
\blacken\path(3527.789,1054.601)(3612.000,964.000)(3580.332,1083.571)(3527.789,1054.601)
\put(462,829){\ellipse{450}{450}}
\put(3837,829){\ellipse{372}{372}}
\path(687,829)(1362,829)
\blacken\path(1242.000,799.000)(1362.000,829.000)(1242.000,859.000)(1242.000,799.000)
\path(1812,829)(2487,829)
\blacken\path(2367.000,799.000)(2487.000,829.000)(2367.000,859.000)(2367.000,799.000)
\path(2937,829)(3612,829)
\blacken\path(3492.000,799.000)(3612.000,829.000)(3492.000,859.000)(3492.000,799.000)
\path(462,604)(462,379)(2712,379)(2712,604)
\blacken\path(2742.000,484.000)(2712.000,604.000)(2682.000,484.000)(2742.000,484.000)
\path(1587,604)(1587,244)(3837,244)(3837,604)
\blacken\path(3867.000,484.000)(3837.000,604.000)(3807.000,484.000)(3867.000,484.000)
\path(12,829)(237,829)
\blacken\path(117.000,799.000)(237.000,829.000)(117.000,859.000)(117.000,799.000)
\put(462,762){\makebox(0,0)[b]{\smash{{\SetFigFont{9}{10.8}{\familydefault}{\mddefault}{\updefault}$1$}}}}
\put(1587,762){\makebox(0,0)[b]{\smash{{\SetFigFont{9}{10.8}{\familydefault}{\mddefault}{\updefault}$2$}}}}
\put(2712,762){\makebox(0,0)[b]{\smash{{\SetFigFont{9}{10.8}{\familydefault}{\mddefault}{\updefault}$3$}}}}
\put(3837,762){\makebox(0,0)[b]{\smash{{\SetFigFont{9}{10.8}{\familydefault}{\mddefault}{\updefault}$4$}}}}
\put(2712,1369){\makebox(0,0)[b]{\smash{{\SetFigFont{8}{9.6}{\familydefault}{\mddefault}{\updefault}$e_2,\ldots,e_6$}}}}
\put(3837,1369){\makebox(0,0)[b]{\smash{{\SetFigFont{8}{9.6}{\familydefault}{\mddefault}{\updefault}$e_1,\ldots,e_6$}}}}
\put(462,1369){\makebox(0,0)[b]{\smash{{\SetFigFont{8}{9.6}{\familydefault}{\mddefault}{\updefault}$e_1,e_2,e_3$}}}}
\put(1587,1369){\makebox(0,0)[b]{\smash{{\SetFigFont{8}{9.6}{\familydefault}{\mddefault}{\updefault}$e_1,e_4,e_5,e_6$}}}}
\put(2262,1954){\makebox(0,0)[b]{\smash{{\SetFigFont{8}{9.6}{\familydefault}{\mddefault}{\updefault}$e_6$}}}}
\put(2712,64){\makebox(0,0)[b]{\smash{{\SetFigFont{8}{9.6}{\familydefault}{\mddefault}{\updefault}$e_3$}}}}
\put(1002,199){\makebox(0,0)[b]{\smash{{\SetFigFont{8}{9.6}{\familydefault}{\mddefault}{\updefault}$e_5$}}}}
\put(1025,874){\makebox(0,0)[b]{\smash{{\SetFigFont{8}{9.6}{\familydefault}{\mddefault}{\updefault}$e_4$}}}}
\put(2172,874){\makebox(0,0)[b]{\smash{{\SetFigFont{8}{9.6}{\familydefault}{\mddefault}{\updefault}$e_2$}}}}
\put(3297,874){\makebox(0,0)[b]{\smash{{\SetFigFont{8}{9.6}{\familydefault}{\mddefault}{\updefault}$e_1$}}}}
\end{picture}
}
\end{center}
\caption[DFA $\cA_4$ with $\kappa(L(\cA_4)) = 4$ and $\sigma(L(\cA_4)) = 24$]{DFA $\cA_4$ with $\kappa(L(\cA_4)) = 4$ and $\sigma(L(\cA_4)) = 24$; the input performing the identity transformation is not shown.}
\label{fig:RTDFA}
\end{figure}
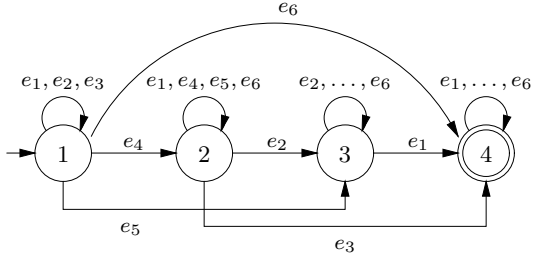

\section{$\gJ$-Trivial Regular Languages}\label{sec:Jtrivial}

For any $m \ge 1$, we define an equivalence relation $\lra_m$ on $\Sig^*$ as follows. 
For any $u, v \in \Sig^*$, $u \mathbin{\lra_m} v$ if any only if for every $x \in \Sig^*$ with $|x| \le m$, 
$x$ is a subword of $u$ if and only if $x$ is a subword of $v$.
Let $L$ be any language over $\Sig$. Then $L$ is \emph{piecewise-testable} if there exists $m \ge 1$ such that, for every $u, v \in \Sig^*$, $u \mathbin{\lra_m} v$ implies that $u \in L \mathbin{\Lra} v \in L$. Let $\cA = (Q, \Sig, \delta, q_1, F)$ be a DFA. If $\Gamma$ is a subset of $\Sig$, a \emph{component} of $\cA$ restricted to $\Gamma$ is a minimal subset $P$ of $Q$ such that, for all $p \in Q$ and $w \in \Gamma^*$, $\delta(p, w) \in P$ if and only of $p \in P$. A state q of $\cA$ is \emph{maximal} if $\delta(q, a) = q$ for all $a \in \Sig$. Simon~\cite{Sim75} proved the following characterization of piecewise-testable languages. 

\begin{theorem}[Simon]\label{thm:simon}
Let $L$ be a regular language over $\Sig$, let $\cA$ be its quotient DFA, and let $T_L$ be its syntactic monoid. Then the following are equivalent:
\be
\item 
$L$ is piecewise-testable.
\item 
$\cA$ is partially ordered, and for every non-empty subset $\Gamma$ of $\Sig$, each component of $\cA$ restricted to $\Gamma$ has exactly one maximal state.
\item 
$T_L$ is $\gJ$-trivial. 
\ee
\end{theorem}

Consequently, a regular language is piecewise-testable if and only if it is $\gJ$-trivial. The following characterization of $\gJ$-trivial monoids is due to Saito~\cite{Sai98}. 

\begin{theorem}[Saito]\label{thm:saito}
Let $S$ be a monoid of transformations of $Q$. Then the following are equivalent:
\be
\item 
$S$ is $\gJ$-trivial.
\item 
$S$ is a subset of $\cF_Q$ and $\Orbit(ts) = \Orbit(t) \join \Orbit(s)$ for all $t, s \in S$. 
\ee
\end{theorem}

\vspace{12pt}
Let $L$ be a  $\gJ$-trivial language with quotient DFA $\cA = (Q, \Sig, \delta, q_1, F)$ and syntactic monoid $T_L$. Since $T_L\subseteq \cF_Q$, an upper bound on the cardinality of $\gJ$-trivial submonoids of~$\cF_Q$ is an upper bound on the syntactic complexity of~$L$. 

\begin{lemma}\label{lem:fixmax} 
If $t,s \in \cF_Q$, then
\be
\item 
$\Fix(t) = \Max(\Orbit(t))$.
\item 
$\Orbit(t) \preceq \Orbit(s)$ implies $\Fix(t) \supseteq \Fix(s)$, where $\Fix(t) = \Fix(s)$ if and only if $\Orbit(t) = \Orbit(s)$.
\ee 
\end{lemma}

\begin{proof} 
1. First, for each $j \in \Max(\Orbit(t))$, since $t \in \cF_Q$, we have $jt = j$, and $j \in \Fix(t)$. So $\Max(\Orbit(t)) \subseteq \Fix(t)$. On the other hand, if there exists $j \in \Fix(t) \setminus \Max(\Orbit(t))$, then $jt = j$, and $j < \max(\omega_t(j))$. Let $i = \max(\omega_t(j))$; then $it = i$ and, for any $k, l \ge 0$, $jt^k = j < i = it^l$. So $i \not\in \omega_t(j)$, which is a contradiction. Hence $\Fix(t) = \Max(\Orbit(t))$. 

2. Assume $\Orbit(t) \preceq \Orbit(s)$. By definition, we have $\Max(\Orbit(t)) \supseteq \Max(\Orbit(s))$. Then, by~1, $\Fix(t) \supseteq \Fix(s)$. Furthermore, $\Orbit(t) = \Orbit(s)$ if and only if $\Max(\Orbit(t)) = \Max(\Orbit(s))$, and if and only if $\Fix(t) = \Fix(s)$. \qed
\end{proof}

\begin{example}\label{ex:fixmax}
Consider non-decreasing $t = [1,3,3,5,6,6]$, as shown in Fig.~\ref{fig:fixmax}~(a). The orbit set $\Orbit(t)$ has three blocks: $\{1\}$, $\{2,3\}$, and $\{4,5,6\}$. Note that $\Fix(t) = \{1,3,6\} = \Max(\Orbit(t))$, as expected. 

Let $s = [4,3,3,6,6,6]$ be another non-decreasing transformation, as shown in Fig.~\ref{fig:fixmax}~(b). The orbit set $\Orbit(s)$ has two blocks: $\{1,4,5,6\}$ and $\{2,3\}$. Note that $\Orbit(t) \prec \Orbit(s)$ and $\Fix(t) \supset \Fix(s)$. \qedb

\begin{figure}[hbt]
\begin{center}
\setlength{\unitlength}{0.00065617in}
\begingroup\makeatletter\ifx\SetFigFont\undefined%
\gdef\SetFigFont#1#2#3#4#5{%
  \reset@font\fontsize{#1}{#2pt}%
  \fontfamily{#3}\fontseries{#4}\fontshape{#5}%
  \selectfont}%
\fi\endgroup%
{\renewcommand{\dashlinestretch}{30}
\begin{picture}(4113,1339)(0,-10)
\put(1577.500,2011.500){\arc{4005.253}{1.0921}{2.0495}}
\blacken\path(2404.842,154.974)(2500.000,234.000)(2378.849,209.052)(2404.842,154.974)
\put(1915.000,508.500){\arc{261.000}{2.3318}{7.0930}}
\blacken\path(2028.393,486.820)(2005.000,414.000)(2054.601,472.221)(2028.393,486.820)
\put(3962.500,508.300){\arc{261.402}{2.3356}{7.0891}}
\blacken\path(4076.335,486.839)(4053.000,414.000)(4102.555,472.261)(4076.335,486.839)
\put(565,324){\ellipse{226}{226}}
\put(1229,330){\ellipse{226}{226}}
\put(1915,324){\ellipse{226}{226}}
\put(2612,324){\ellipse{226}{226}}
\put(3287,324){\ellipse{226}{226}}
\put(3963,324){\ellipse{226}{226}}
\path(1330,324)(1780,324)
\blacken\path(1705.000,309.000)(1780.000,324.000)(1705.000,339.000)(1705.000,309.000)
\path(2725,324)(3175,324)
\blacken\path(3100.000,309.000)(3175.000,324.000)(3100.000,339.000)(3100.000,309.000)
\path(3400,324)(3850,324)
\blacken\path(3775.000,309.000)(3850.000,324.000)(3775.000,339.000)(3775.000,309.000)
\put(1915.000,1183.500){\arc{261.000}{2.3318}{7.0930}}
\blacken\path(2028.393,1161.820)(2005.000,1089.000)(2054.601,1147.221)(2028.393,1161.820)
\put(3962.500,1183.300){\arc{261.402}{2.3356}{7.0891}}
\blacken\path(4076.335,1161.839)(4053.000,1089.000)(4102.555,1147.261)(4076.335,1161.839)
\put(565.000,1183.500){\arc{261.000}{2.3318}{7.0930}}
\blacken\path(678.393,1161.820)(655.000,1089.000)(704.601,1147.221)(678.393,1161.820)
\put(3265.000,2101.500){\arc{2656.525}{1.1147}{2.0269}}
\blacken\path(3788.058,864.131)(3850.000,909.000)(3775.566,891.406)(3788.058,864.131)
\put(565,999){\ellipse{226}{226}}
\put(1229,1005){\ellipse{226}{226}}
\put(1915,999){\ellipse{226}{226}}
\put(2612,999){\ellipse{226}{226}}
\put(3287,999){\ellipse{226}{226}}
\put(3963,999){\ellipse{226}{226}}
\path(1330,999)(1780,999)
\blacken\path(1705.000,984.000)(1780.000,999.000)(1705.000,1014.000)(1705.000,984.000)
\path(3400,999)(3850,999)
\blacken\path(3775.000,984.000)(3850.000,999.000)(3775.000,1014.000)(3775.000,984.000)
\put(565,279){\makebox(0,0)[b]{\smash{{\SetFigFont{7}{8.4}{\familydefault}{\mddefault}{\updefault}$1$}}}}
\put(1227,279){\makebox(0,0)[b]{\smash{{\SetFigFont{7}{8.4}{\familydefault}{\mddefault}{\updefault}$2$}}}}
\put(1915,279){\makebox(0,0)[b]{\smash{{\SetFigFont{7}{8.4}{\familydefault}{\mddefault}{\updefault}$3$}}}}
\put(2612,279){\makebox(0,0)[b]{\smash{{\SetFigFont{7}{8.4}{\familydefault}{\mddefault}{\updefault}$4$}}}}
\put(3287,279){\makebox(0,0)[b]{\smash{{\SetFigFont{7}{8.4}{\familydefault}{\mddefault}{\updefault}$5$}}}}
\put(3963,279){\makebox(0,0)[b]{\smash{{\SetFigFont{7}{8.4}{\familydefault}{\mddefault}{\updefault}$6$}}}}
\put(115,279){\makebox(0,0)[b]{\smash{{\SetFigFont{7}{8.4}{\familydefault}{\mddefault}{\updefault}(b)}}}}
\put(115,954){\makebox(0,0)[b]{\smash{{\SetFigFont{7}{8.4}{\familydefault}{\mddefault}{\updefault}(a)}}}}
\put(565,954){\makebox(0,0)[b]{\smash{{\SetFigFont{7}{8.4}{\familydefault}{\mddefault}{\updefault}$1$}}}}
\put(1227,954){\makebox(0,0)[b]{\smash{{\SetFigFont{7}{8.4}{\familydefault}{\mddefault}{\updefault}$2$}}}}
\put(1915,954){\makebox(0,0)[b]{\smash{{\SetFigFont{7}{8.4}{\familydefault}{\mddefault}{\updefault}$3$}}}}
\put(2612,954){\makebox(0,0)[b]{\smash{{\SetFigFont{7}{8.4}{\familydefault}{\mddefault}{\updefault}$4$}}}}
\put(3287,954){\makebox(0,0)[b]{\smash{{\SetFigFont{7}{8.4}{\familydefault}{\mddefault}{\updefault}$5$}}}}
\put(3963,954){\makebox(0,0)[b]{\smash{{\SetFigFont{7}{8.4}{\familydefault}{\mddefault}{\updefault}$6$}}}}
\end{picture}
}
\end{center}
\caption{Nondecreasing transformations $t = [1,3,3,5,6,6]$ and $s = [4,3,3,6,6,6]$.}
\label{fig:fixmax}
\end{figure}
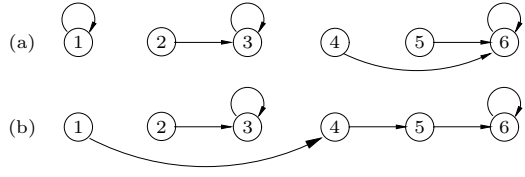
\end{example}

\newcommand{\trmax}{{t_{\mathrm{max}}}} 

Define the transformation $\trmax = [2,3,\ldots,n,n]$. The subscript ``$\mathrm{max}$'' is chosen because $\Orbit(\trmax) = \{Q\}$ is the maximal element in the lattice $\rmPi_Q$. Clearly $\trmax \in \cF_Q$ and $\Fix(\trmax) = \{n\}$. For any submonoid $S$ of $\cF_Q$, let $S[\trmax]$ be the smallest monoid containing $\trmax$ and all elements of $S$. 

\begin{lemma}\label{lem:trmax} 
Let $S$ be a $\gJ$-trivial submonoid of $\cF_Q$. Then
\be 
\item 
$S[\trmax]$ is $\gJ$-trivial.
\item 
Let $\cA = (Q, \Sig, \delta, 1, \{n\})$ be the DFA in which each $a \in \Sig$ defines a distinct transformation in $S[\trmax]$. Then $\cA$ is minimal. 
\ee 
\end{lemma}

\begin{proof}
1. By Theorem~\ref{thm:saito}, it is enough to prove that for any $t \in S$, $\Orbit(t) \join \Orbit(\trmax) = \Orbit(t \trmax)$ and $\Orbit(\trmax) \join \Orbit(t) = \Orbit(\trmax t)$. Note that $\Orbit(\trmax) = \{Q\}$; so we have $\Orbit(t) \join \Orbit(\trmax) = \Orbit(\trmax) \join \Orbit(t) = \{Q\}$. On the other hand, since $S \subseteq \cF_Q$ and $\trmax \in \cF_Q$, both $t \trmax$ and $\trmax t$ are non-decreasing as well. Suppose $i \in \Fix(t \trmax)$; then $i(t \trmax) = (it)\trmax = i$. Since $\trmax$ is non-decreasing, $it \le i$; and since $t$ is also non-decreasing, $i \le it$. Hence $it = i$, and $i \trmax = i$, which implies that $i \in \Fix(\trmax)$ and $i = n$. Then $\Fix(t \trmax) = \{n\}$ and $\Orbit(t \trmax) = \{Q\}$. Similarly, $\Fix(\trmax t) = \{n\}$ and $\Orbit(\trmax t) = \{Q\}$. Therefore $S[\trmax]$ is also $\gJ$-trivial.

2. Suppose $a_0 \in \Sig$ performs the transformation $\trmax$. Each state $p \in Q$ can be reached from the initial state 1 by the word $u = a_0^{p-1}$, and $p$ accepts the word $v = a_0^{n-p}$, while all other states reject $v$. So $\cA$ is minimal. \qed
\end{proof}

For any $\gJ$-trivial submonoid $S$ of $\cF_Q$, we denote by $\cA(S, \trmax)$ the DFA in Lemma~\ref{lem:trmax}. Then $\cA(S, \trmax)$ is the quotient DFA of some $\gJ$-trivial regular language $L$. Next, we have

\begin{lemma}\label{lem:orbitfix} 
Let $S$ be a $\gJ$-trivial submonoid of $\cF_Q$. For any $t, s \in S$, 
if $\Fix(t) = \Fix(s)$, then $\Orbit(t) = \Orbit(s)$. 
\end{lemma}

\begin{proof}
Pick any $t, s \in S$ such that $\Fix(t) = \Fix(s)$. If $t = s$, then it is trivial that $\Orbit(t) = \Orbit(s)$. Assume $t \neq s$, and $\Orbit(t) \neq \Orbit(s)$. By Part~2 of Lemma~\ref{lem:fixmax}, we have $\Orbit(t) \not\prec \Orbit(s)$ and $\Orbit(s) \not\prec \Orbit(t)$. Then there exists $i \in Q$ such that $\omega_t(i) \not\subseteq \omega_s(i)$. Let $p = \max(\omega_t(i))$. We define $q \in Q$ as follows. If $\max(\omega_t(i)) \neq \max(\omega_s(i))$, then let $q = \max(\omega_s(i))$; so $q \neq p$. Otherwise $\max(\omega_t(i)) = \max(\omega_s(i))$, and there exists $j \in \omega_t(i)$ such that $j \not\in \omega_s(i)$; let $q = \max(\omega_s(j))$. Now $p = \max(\omega_t(j)) = \max(\omega_t(i)) = \max(\omega_s(i))$, and since $j \not\in \omega_s(i)$, we have $q \neq p$ as well. Note that $p,q \in \Fix(t) = \Fix(s)$ in both cases. Consider the DFA $\cA(S, \trmax)$ with alphabet $\Sig$, and suppose that $a \in \Sig$ performs $t$ and $b \in \Sig$ performs~$s$. Let $\cB$ be the DFA $\cA(S, \trmax)$ restricted to $\{a,b\}$. Since $p \in \omega_t(i)$ and $q \in \omega_s(i)$, $p, q$ are in the same component $P$ of $\cB$. However, $p$ and $q$ are two distinct maximal states in $P$, which contradicts Theorem~\ref{thm:simon}. Therefore $\Orbit(t) = \Orbit(s)$. \qed
\end{proof}

\begin{example}\label{ex:orbitfix}
To illustrate one usage of Lemma~\ref{lem:orbitfix}, we consider two non-decreasing transformations $t = [2,2,4,4]$ and $s = [3,2,4,4]$. They have the same set of fixed points $\Fix(t) = \Fix(s) = \{2,4\}$. However, $\Orbit(t) = \{\{1,2\}, \{3,4\}\}$ and $\Orbit(s) = \{\{2\}, \{1,3,4\}\}$. By Lemma~\ref{lem:orbitfix}, $t$ and $s$ cannot appear together in a $\gJ$-trivial monoid. Indeed, consider any minimal DFA $\cA$ having at least two inputs $a,b$ such that $a$ performs $t$ and $b$ performs~$s$. The DFA $\cB$ of $\cA$ restricted to the alphabet $\{a,b\}$ is shown in Fig.~\ref{fig:orbitfix}. There is only one component in $\cB$, but there are two maximal states $2$ and $4$. By Theorem~\ref{thm:simon}, the syntactic monoid of $\cA$ is not $\gJ$-trivial. \qedb
\end{example}

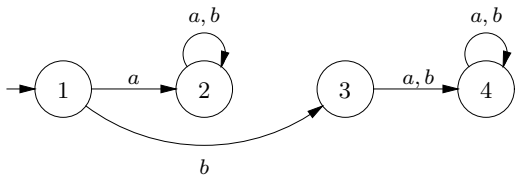
\begin{figure}[hbt]
\begin{center}
\setlength{\unitlength}{0.00065617in}
\begingroup\makeatletter\ifx\SetFigFont\undefined%
\gdef\SetFigFont#1#2#3#4#5{%
  \reset@font\fontsize{#1}{#2pt}%
  \fontfamily{#3}\fontseries{#4}\fontshape{#5}%
  \selectfont}%
\fi\endgroup%
{\renewcommand{\dashlinestretch}{30}
\begin{picture}(4070,1408)(0,-10)
\put(1587.000,1011.250){\arc{337.500}{2.4981}{6.9267}}
\blacken\path(1725.575,1033.641)(1722.000,910.000)(1783.339,1017.413)(1725.575,1033.641)
\put(1587,730){\ellipse{450}{450}}
\put(3837.000,1011.250){\arc{337.500}{2.4981}{6.9267}}
\blacken\path(3975.575,1033.641)(3972.000,910.000)(4033.339,1017.413)(3975.575,1033.641)
\put(3837,730){\ellipse{450}{450}}
\put(1587.000,1855.000){\arc{3150.000}{0.9273}{2.2143}}
\blacken\path(2450.498,501.955)(2532.000,595.000)(2416.300,551.256)(2450.498,501.955)
\put(462,730){\ellipse{450}{450}}
\put(2712,730){\ellipse{450}{450}}
\path(2937,730)(3612,730)
\blacken\path(3492.000,700.000)(3612.000,730.000)(3492.000,760.000)(3492.000,700.000)
\path(12,730)(237,730)
\blacken\path(117.000,700.000)(237.000,730.000)(117.000,760.000)(117.000,700.000)
\path(687,730)(1362,730)
\blacken\path(1242.000,700.000)(1362.000,730.000)(1242.000,760.000)(1242.000,700.000)
\put(462,663){\makebox(0,0)[b]{\smash{{\SetFigFont{9}{10.8}{\familydefault}{\mddefault}{\updefault}$1$}}}}
\put(1587,663){\makebox(0,0)[b]{\smash{{\SetFigFont{9}{10.8}{\familydefault}{\mddefault}{\updefault}$2$}}}}
\put(2712,663){\makebox(0,0)[b]{\smash{{\SetFigFont{9}{10.8}{\familydefault}{\mddefault}{\updefault}$3$}}}}
\put(3837,663){\makebox(0,0)[b]{\smash{{\SetFigFont{9}{10.8}{\familydefault}{\mddefault}{\updefault}$4$}}}}
\put(1025,775){\makebox(0,0)[b]{\smash{{\SetFigFont{8}{9.6}{\familydefault}{\mddefault}{\updefault}$a$}}}}
\put(1587,55){\makebox(0,0)[b]{\smash{{\SetFigFont{8}{9.6}{\familydefault}{\mddefault}{\updefault}$b$}}}}
\put(1587,1270){\makebox(0,0)[b]{\smash{{\SetFigFont{8}{9.6}{\familydefault}{\mddefault}{\updefault}$a,b$}}}}
\put(3837,1270){\makebox(0,0)[b]{\smash{{\SetFigFont{8}{9.6}{\familydefault}{\mddefault}{\updefault}$a,b$}}}}
\put(3297,775){\makebox(0,0)[b]{\smash{{\SetFigFont{8}{9.6}{\familydefault}{\mddefault}{\updefault}$a,b$}}}}
\end{picture}
}
\end{center}
\caption{DFA $\cB$ with two inputs $a$ and $b$, where $t_a = [2,2,4,4]$ and $t_b = [3,2,4,4]$.}
\label{fig:orbitfix}
\end{figure}

Let $\pi$ be any partition of $Q$. A block $X$ of $\pi$ is \emph{trivial} if it contains only one element of $Q$; otherwise it is \emph{non-trivial}. We define the set $\cE(\pi) = \{t \in \cF_Q \mid \Orbit(t) = \pi\}$. Then 

\begin{lemma}\label{lem:Ecard} 
If $\pi$ is a partition of $Q$ with $r$ blocks, where $1 \le r \le n$, then $|\cE(\pi)| \le (n-r)!$. Moreover, the equality holds if and only if $\pi$ has exactly one non-trivial block.
\end{lemma}

\begin{proof}
Suppose $\pi = \{X_1, \ldots, X_r\}$, and $|X_i| = k_i$ for each $i$, $1 \le i \le r$. Without loss of generality, we can rearrange blocks $X_i$ so that $k_1 \le \cdots \le k_r$. Let $t \in \cE(\pi)$ be any transformation. Then $t \in \cF_Q$, and hence $\Fix(t) = \Max(\Orbit(t)) = \Max(\pi)$. Consider each block $X_i$, and suppose $X_i = \{j_1, \ldots, j_{k_i}\}$ with $j_1 < \cdots < j_{k_i}$. Since $j_{k_i} = \max(X_i)$, we have $j_{k_i} \in \Fix(t)$ and $j_{k_i}t = j_{k_i}$. On the other hand, if $1 \le l < k_i$, then $j_l \not\in \Max(\pi)$, and since $t \in \cF_Q$, we have $j_l t > j_l$; since $j_l t \in \omega_t(j_l) = X_i$, $j_l t \in \{j_{l+1},\ldots,j_{k_i}\}$. So there are $(k_i-1)!$ different $t|_{X_i}$, and there are $\prod_{i=1}^r (k_i-1)!$ different transformations $t$ in $\cE(\pi)$. 

Clearly, if $r = 1$, then $k_r = n$ and $|\cE(\pi)| = (n-1)!$. Assume $r \ge 2$. Note that $k_i \ge 1$ for all $i$, $1 \le i \le r$, and $\sum_{i=1}^r k_i = n$. If $k_1 = \cdots = k_{r-1} = 1$, then $k_r = n-r+1$, and $|\cE(\pi)| = (k_r - 1)! \prod_{i=1}^{r-1}0! = (n-r)!$. Otherwise, let $h$ be the smallest index such that $k_h > 1$. For all $i$, $h \le i \le r-1$, since $k_i \le k_r$, we have $(k_i-1)! < (k_i-1)^{k_i-1} \le (k_r-1)^{k_i-1}$. Then
\begin{align*}
  |\cE(\pi)| &= (k_r-1)! \prod_{i=1}^{h-1}0! \prod_{i=h}^{r-1} (k_i-1)! 
             < (k_r-1)! \prod_{i=h}^{r-1} (k_r-1)^{k_i-1} \\
             &= (k_r-1)! \cdot (k_r-1)^{\sum_{i=h}^{r-1}(k_i-1)} \\
             &< (k_r-1)! \cdot k_r (k_r+1) \cdots (k_r-1+\sum_{i=h}^{r-1}(k_i-1)) \\
             &= (k_r-1)! \cdot k_r (k_r+1) \cdots (n-r) 
             = (n-r)!
\end{align*}
Therefore the lemma holds. \qed
\end{proof}

\begin{example}\label{ex:Ecard}
Suppose $n = 10$, $r = 3$, and consider the partition $\pi = \{X_1, X_2, X_3\}$, where $X_1 = \{1,2,5\}$, $X_2 = \{3,7\}$, and $X_3 = \{4,6,8,9,10\}$. Then $k_1 = |X_1| = 3$, $k_2 = |X_2| = 2$, and $k_3 = |X_3| = 5$. Let $t \in \cE(\pi)$ be an arbitrary transformation; then $\Fix(t) = \{5,7,10\}$. For any $i \in X_1$, if $i = 1$, then $it$ could be $2$ or $5$; otherwise $i = 2$ or $5$, and $it$ must be $5$. So there are $(k_1-1)! = 2!$ different $t|_{X_1}$. Similarly, there are $(k_2-1)! = 1!$ different $t|_{X_2}$ and $(k_3-1)! = 4!$ different $t|_{X_3}$. So $|\cE(\pi)| = 2!1!4! = 48$. 

Consider another partition $\pi' = \{X_1', X_2', X_3'\}$ with three blocks, where $X_1' = \{5\}$, $X_2' = \{7\}$, and $X_3' = \{1,2,3,4,6,8,9,10\}$. Now $k_1 = |X_1'| = 1$, $k_2 = |X_2'| =~1$, and $k_3 = |X_3'| = 8$. We have $\Max(\pi') = \Max(\pi) = \{5,7,10\}$. Then, for any $t \in \cE(\pi')$, $\Fix(t) = \{5,7,10\}$ as well. Since $k_1 = k_2 = 1$, both $t|_{X_1}$ and $t|_{X_2}$ are unique. There are $(k_3-1)! = 7!$ different $t|_{X_3}$. Together we have $|\cE(\pi')| = 1!1!7! = (10-3)! = 5040$, which is the upper bound in Lemma~\ref{lem:Ecard} for $n = 10$ and $r = 3$. \qedb
\end{example}

Note that, for any $t \in \cF_Q$, we have $n \in \Fix(t)$. Let $\cP_n(Q)$ be the set of all subsets $Z$ of $Q$ such that $n \in Z$. Then we obtain the following upper bound. 

\begin{proposition}\label{prop:Jbound} 
For $n \ge 1$, if $S$ is a $\gJ$-trivial submonoid of $\cF_Q$, then $$|S| \le \sum_{r=1}^n C^{n-1}_{r-1} (n-r)! = \lfloor e(n-1)! \rfloor.$$ 
\end{proposition}

\begin{proof} 
Assume $S$ is a $\gJ$-trivial submonoid of $\cF_Q$. For any $Z \in \cP_n(Q)$, let $S_Z = \{t \in S \mid \Fix(t) = Z\}$. Then $S = \bigcup_{Z \in \cP_n(Q)} S_Z$, and for any $Z_1, Z_2 \in \cP_n(Q)$ with $Z_1 \neq Z_2$, $S_{Z_1} \cap S_{Z_2} = \emptyset$.

Pick any $Z \in \cP_n(Q)$. By Lemma~\ref{lem:orbitfix}, for any $t, s \in S_Z$, since $\Fix(t) = \Fix(s) = Z$, we have $\Orbit(t) = \Orbit(s) = \pi$ for some partition $\pi \in \rmPi_Q$. Then $S_Z \subseteq \cE(\pi)$. Suppose $r = |Z|$. By Lemma~\ref{lem:Ecard}, $|S_Z| \le |\cE(\pi)| \le (n-r)!$. Since $n \in Z$, $1 \le r \le n$; and since there are $C^{n-1}_{r-1}$ different $Z \in \cP_n(Q)$, we have 
 $$|S| = \sum_{Z \in \cP_n(Q)}|S_Z|  \le \sum_{r=1}^n C^{n-1}_{r-1} (n-r)! 
                              =  \sum_{r=1}^n \frac{(n-1)!}{(r-1)!}
                             =  \lfloor e(n-1)! \rfloor. $$
The last equality is a well-known  identity in combinatorics. \qed
\end{proof}

The above upper bound is met by the following monoid $\cS_n$. For any $Z \in \cP_n(Q)$, suppose $Z = \{j_1, \ldots, j_r,n\}$ such that $j_1 < \cdots < j_r < n$ for some $r \ge 0$; then we define partition $\pi_Z = \{Q\}$ if $Z = \{n\}$, and $\pi_Z = \{\{j_1\}, \ldots, \{j_r\}, Q \setminus \{j_1,\ldots,j_r\}\}$ otherwise. Let 
$$\cS_n = \bigcup_{Z \in \cP_n(Q)} \cE(\pi_Z).$$

\begin{example}\label{ex:cS} 
Suppose $n = 4$; then $|\cP_4(Q)| = 2^3 = 8$. First consider $Z = \{1,3,4\} \in \cP_4(Q)$. By definition, $\pi_Z = \{\{1\}, \{3\}, \{2,4\}\}$. There is only one transformation $t_1 = [1,4,3,4]$ in $\cE(\pi_Z)$. If $Z' = \{3,4\}$, then $\pi_{Z'} = \{\{3\}, \{1,2,4\}$. There are two transformations $t_2 = [2,4,3,4]$ and $t_3 = [4,4,3,4]$ in $\cE(\pi_{Z'})$. Table~\ref{tab:cS} summarizes the number of transformations in $\cE(\pi_Z)$ for each $Z \in \cP_4(Q)$. Note that the set $\cS_4$ contains $16$ transformations in total. 
\qedb

\begin{table}[hbt]
\caption{Number of transformations in $\cE(\pi_Z)$ for each $Z \in \cP_4(Q)$.}
\label{tab:cS}
\begin{center}
$
\begin{array}{|c|c|c|}    
\hline
 ~Z~ & ~\txt{Blocks of} \pi_Z~ & ~|\cE(\pi_Z)|~ \\ 
\hline \hline
 \{1,2,3,4\} & \{1\}, \{2\}, \{3\}, \{4\} & 1 \\
\hline 
 \{1,2,4\}   & \{1\}, \{2\}, \{3,4\}      & 1 \\
\hline 
 \{1,3,4\}   & \{1\}, \{3\}, \{2,4\}      & 1 \\
\hline  
 \{2,3,4\}   & \{2\}, \{3\}, \{1,4\}      & 1 \\
\hline 
 \{1,4\}     & \{1\}, \{2,3,4\}           & 2 \\
\hline 
 \{2,4\}     & \{2\}, \{1,3,4\}           & 2 \\
\hline 
 \{3,4\}     & \{3\}, \{1,2,4\}           & 2 \\
\hline 
 \{4\}       & \{1,2,3,4\}                & 6 \\
\hline
\end{array}
$
\end{center}
\end{table}
\end{example}

\begin{proposition}\label{prop:cS}
For $n \ge 1$, the set $\cS_n$ is a $\gJ$-trivial submonoid of $\cF_Q$ with cardinality 
\begin{equation}\label{eq:gt}
  g(n) = |\cS_n| = \sum_{r=1}^n C^{n-1}_{r-1} (n-r)! = \lfloor e(n-1)! \rfloor. 
\end{equation}
\end{proposition}

\begin{proof}
First we prove the following claim: 

\medskip

{\bf Claim:} For any $t, s \in \cS_n$, $\Orbit(ts) = \pi_Z$ for some $Z \in \cP_n(Q)$. 

Let $t \in \cE(\pi_{Z_1})$ and $s \in \cE(\pi_{Z_2})$ for some $Z_1, Z_2 \in \cP_n(Q)$. Suppose $\Orbit(ts) \neq \pi_Z$ for any $Z \in \cP_n(Q)$. Then there exists a block $X_0 \in \Orbit(ts)$ such that $n \not\in X_0$ and $|X_0| \ge 2$. Suppose $i \in X_0$ with $i \neq \max(X_0)$. We must have $i \in \omega_t(n)$ or $it \in \omega_s(n)$; otherwise $it = i$ and $(it)s = it = i$, which implies $i = \max(X_0)$. However, in either case, there exists large $m$ such that $it^m = n$ or $i(ts)^m = n$, respectively. Then $n \in \omega_{ts}(i) = X_0$, a contradiction. So the claim holds. \qedb 

\medskip

By the claim, for any $t, s \in \cS_n$, since $\Orbit(ts) = \pi_Z$ for some $Z \in \cP_n(Q)$, $ts \in \cE(\pi_Z) \subseteq \cS_n$. Hence $\cS_n$ is a submonoid of $\cF_Q$. 

Next we show that $\cS_n$ is $\gJ$-trivial. Pick any $t, s \in \cS_n$, and suppose $t \in \cE(\pi_{Z_1})$ and $s \in \cE(\pi_{Z_2})$ for some $Z_1, Z_2 \in \cP_n(Q)$. Suppose $\Max(Z_1) \cap \Max(Z_2) = \{j_1,\ldots,j_r,n\}$, for some $r \ge 0$. Then we have $Z_1 \join Z_2 = \{\{j_1\}, \ldots, \{j_r\}, X\}$, where $X = Q \setminus \{j_1,\ldots,j_r\}$ and $n \in X$. On the other hand, by the claim, $\Orbit(ts) = \{\{p_1\}, \ldots, \{p_k\}, Y\}$ for some $p_1,\ldots,p_k \in Q$, where $Y = Q \setminus \{p_1,\ldots,p_k\}$ and $n \in Y$. Note that, since $\cS_n \subseteq \cF_Q$, $\Max(\Orbit(ts)) = \Fix(ts) = \Fix(t) \cap \Fix(s) = \Max(Z_1) \cap \Max(Z_2)$. Then $r = k$ and $\{j_1,\ldots,j_r\} = \{p_1,\ldots,p_k\}$. Hence $\Orbit(t) \join \Orbit(s) = Z_1 \join Z_2 = \Orbit(ts)$. By Theorem~\ref{thm:saito}, $\cS_n$ is $\gJ$-trivial. 

For any $Z \in \cP_n(Q)$ with $|Z| = r$, where $1 \le r \le n$, we have $\pi_Z = \{X_1,\ldots,X_r\}$ with $k_i = |X_i| = 1$ for $1 \le i < r$, and $k_r = |X_r|$. By Lemma~\ref{lem:Ecard}, $|\cE(\pi_Z)| = (n-r)!$. Moreover, if $Z_1 \neq Z_2$, then $\cE(\pi_{Z_1}) \cap \cE(\pi_{Z_2}) = \emptyset$. Since $n \in Z$ is fixed, there are $C^{n-1}_{r-1}$ different $Z$. Therefore $|\cS_n| = \sum_{r=1}^n C^{n-1}_{r-1} (n-r)! = \lfloor e(n-1)! \rfloor$. \qed
\end{proof}

We now define a generating set of the monoid $\cS_n$. 
Suppose $n \ge 1$. For any $Z \in \cP_n(Q)$, if $Z = Q$, then let $t_Z = \tid_Q$. Otherwise, let $h_Z = \max(Q \setminus Z)$, and let $t_Z$ be a transformation of $Q$ defined by: For all $i \in Q$, 
\begin{equation*}
it \defeq \begin{cases} i & \txt{if} i \in Z, \\
n & \txt{if} i = h_Z, \\
h_Z & \txt{otherwise.}
\end{cases}
\end{equation*}
Let $\cG\cS_n = \{t_Z \mid Z \in \cP_n(Q)\}$. 

\begin{example}\label{ex:genJT}
Suppose $n = 5$. As the first example, consider $Z = \{1,3,4,5\}$. Then $h_Z = \max(Q \setminus Z) = 2$, and $t_Z = [1,5,3,4,5]$. If $Z' = \{4,5\}$, then $h_{Z'} = 4$ and $t_{Z'} = [3,3,5,4,5]$. If $Z'' = \{5\}$, then $h_{Z''} = 4$ and $t_{Z''} = [4,4,4,5,5]$. The set $\cG\cS_5$ contains the following $16$ transformations: 
$$\begin{array}{lll}
t_1 = [ 1, 2, 3, 4, 5 ], \quad & \quad t_2 = [ 1, 2, 3, 5, 5 ], \quad & \quad t_3 = [ 1, 2, 4, 5, 5 ], \\
t_4 = [ 1, 2, 5, 4, 5 ], \quad & \quad t_5 = [ 1, 3, 5, 4, 5 ], \quad & \quad t_6 = [ 1, 4, 3, 5, 5 ], \\
t_7 = [ 1, 4, 4, 5, 5 ], \quad & \quad t_8 = [ 1, 5, 3, 4, 5 ], \quad & \quad t_9 = [ 2, 5, 3, 4, 5 ], \\
t_{10} = [ 3, 2, 5, 4, 5 ], \quad & \quad t_{11} = [ 3, 3, 5, 4, 5 ], \quad & \quad t_{12} = [ 4, 2, 3, 5, 5 ], \\
t_{13} = [ 4, 2, 4, 5, 5 ], \quad & \quad t_{14} = [ 4, 4, 3, 5, 5 ], \quad & \quad t_{15} = [ 4, 4, 4, 5, 5 ], \\
t_{16} = [ 5, 2, 3, 4, 5 ].
\end{array}$$ 
\qedb
\end{example}

\begin{proposition}\label{prop:genJT}
For $n \ge 1$, the monoid $\cS_n$ can be generated by the set $\cG\cS_n$ of $2^{n-1}$ transformations of $Q$. 
\end{proposition}

\begin{proof}
First, for any $t_Z \in \cG\cS_n$, where $Z \in \cP_n(Q)$, we have $\Orbit(t_Z) = \pi_Z$; hence $t_Z \in \cE(\pi_Z) \subseteq \cS_n$. So $\cG\cS_n \subseteq \cS_n$ and $\langle \cG\cS_n \rangle \subseteq \cS_n$. 

Fix arbitrary $Z \in \cP_n(Q)$, and suppose $U = Q \setminus Z$. If $U = \emptyset$, then $\pi_Z = \{\{1\},\ldots,\{n\}\}$ and $\cE(\pi_Z) = \{\tid_Q\} \subseteq \langle \cG\cS_n \rangle$. Assume $U \neq \emptyset$ in the following. Let $Y$ be the only non-trivial block in $\pi_Z$. Note that $Y = U \cup \{n\}$ and $h_Z = \max(U)$. For any $t \in \cE(\pi_Z)$, since $\Fix(t) = Z$ and $h_Z \not\in Z$, $h_Zt > h_Z$; and since $Y$ is an orbit of $t$, $h_Zt = n$. We prove by induction on $|U| = |Q \setminus Z|$ that $\cE(\pi_Z) \subseteq \langle \cG\cS_n \rangle$. 
\be
\item If $U = \{h_Z\}$, then $Y = \{h_Z,n\}$. So $h_Zt = n$, and $t = {h_Z \choose n} = t_Z \subseteq \langle \cG\cS_n \rangle$. 

\item Otherwise $U = \{h_1,\ldots,h_l,h_Z\}$ for some $h_1 < \cdots < h_l < h_Z < n$ and $l \ge 1$. Assume that, for any $Z' \in \cP_n(Q)$ with $|Q \setminus Z'| \le l$, we have $\cE(\pi_{Z'}) \subseteq \langle \cG\cS_n \rangle$. Then $Y = \{h_1,\ldots,h_l,h_Z,n\}$, and $t_Z = {h_Z \choose n}{h_l \choose h_Z} \cdots {h_1 \choose h_Z}$. For any $t \in \cE(\pi_Z)$, since $Y$ is an orbit of $t$ and $Q \setminus Y \subseteq \Fix(t)$, $t$ must have the form $t = {h_Z \choose n}{h_l \choose j_l}\cdots {h_1 \choose j_1},$ where $j_{i} \in \{h_{i+1},\ldots,h_l,h_Z,n\}$ for $i = 1,\ldots,l$. 
Let $\{h_1,\ldots,h_l\} = V \cup W$ such that $h_i \in V$ if and only if $j_i = h_it = h_Z$. Suppose $V = \{h_{p_1},\ldots,h_{p_k}\}$ and $W = \{h_{q_1},\ldots,h_{q_m}\}$, where $h_{p_1} < \cdots < h_{p_k}$, $h_{q_1} < \cdots < h_{q_m}$, $0 \le k,m \le l$ and $l = k + m$. Let $t_1 = {h_Z \choose n}$, $t_2 = {h_Z \choose n}{h_{p_1} \choose h_Z}\cdots{h_{p_k} \choose h_Z}$, and $t_3 = {h_{p_1} \choose n}\cdots{h_{p_k} \choose n}{h_{q_1} \choose j_{q_1}}\cdots{h_{q_m} \choose j_{q_m}}$. Note that $t_1 = t_{Z'}$ for $Z' = Q \setminus \{h_Z\}$, and $t_2 = t_{Z''}$ for $Z'' = Q \setminus \{h_{p_1},\ldots,h_{p_k},h_Z\}$. Also note that $\Fix(t_3) = \Fix(t) \cup \{h_Z\}$, and since $j_{q_i} = h_{q_i}t \in U \setminus \{h_Z\}$ for all $h_{q_i} \in W$, we have $t_3 \in \cE(\pi_{Z'''})$ for $Z''' = Z \cup \{h_Z\}$. By assumption, $t_3 \in \langle \cG\cS_n \rangle$. Now
\begin{eqnarray*}
  t_1t_2t_3 &=& {h_Z \choose n} {h_Z \choose n}{h_{p_1} \choose h_Z}\cdots{h_{p_k} \choose h_Z} {h_{p_1} \choose n}\cdots{h_{p_k} \choose n}{h_{q_1} \choose j_{q_1}}\cdots{h_{q_m} \choose j_{q_m}} \\
           &=& {h_Z \choose n}{h_{p_1} \choose h_Z}\cdots{h_{p_k} \choose h_Z} {h_{q_1} \choose j_{q_1}}\cdots{h_{q_m} \choose j_{q_m}}
           = t.
\end{eqnarray*}
Thus $t \in \langle \cG\cS_n \rangle$ and $\cE(\pi_Z) \subseteq \langle \cG\cS_n \rangle$. 
\ee
By induction, $\cS_n = \bigcup_{Z \in \cP_n(Q)} \cE(\pi_Z) \subseteq \langle \cG\cS_n \rangle$. Therefore $\cS_n = \langle \cG\cS_n \rangle$. Since there are $2^{n-1}$ different $Z \in \cP_n(Q)$, there are $2^{n-1}$ transformations in~$\cG\cS_n$. \qed
\end{proof}

\begin{example}\label{ex:GS}
Suppose $n = 5$. The list of all transformations in $\cG\cS_5$ is shown in Example~\ref{ex:genJT}. 
Consider $Z = \{3,5\} \in \cP_5(Q)$, and $t = [2,4,3,5,5] \in \cE(\pi_Z)$. The transition graph of $t$ is shown in Fig.~\ref{fig:genJT}~(a). As in Proposition~\ref{prop:genJT}, we have $U = \{1,2,4\}$ and $h_Z = 4$. To show that $t \in \langle \cG\cS_5 \rangle$, we find $V = \{2\}$ and $W = \{1\}$. Then let $t_1 = {4 \choose 5}$, $t_2 = {4 \choose 5}{2 \choose 4}$, and $t_3 = {2 \choose 5}{1 \choose 2}$. We assume that $t_3 \in \langle \cG\cS_5 \rangle$; in fact, $t_3 = t_{Z'''}$ for $Z''' = \{3,4,5\}$ in this example. The transition graphs of $t_1$, $t_2$, and $t_3$ are shown in Fig.~\ref{fig:genJT}~(b), (c), and~(d), respectively. One can verify that $t = t_1t_2t_3$, and hence $t \in \langle \cG\cS_5 \rangle$. 
\qedb

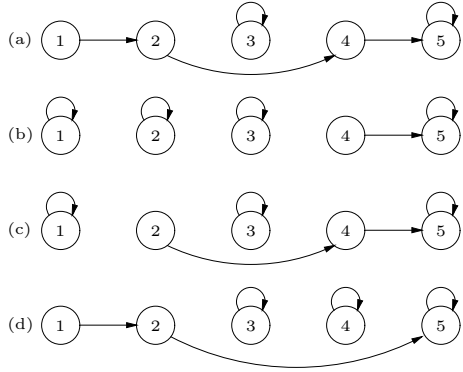
\begin{figure}[hbt]
\begin{center}
\setlength{\unitlength}{0.00043745in}
\begingroup\makeatletter\ifx\SetFigFont\undefined%
\gdef\SetFigFont#1#2#3#4#5{%
  \reset@font\fontsize{#1}{#2pt}%
  \fontfamily{#3}\fontseries{#4}\fontshape{#5}%
  \selectfont}%
\fi\endgroup%
{\renewcommand{\dashlinestretch}{30}
\begin{picture}(5378,4355)(0,-10)
\put(645,3881){\ellipse{450}{450}}
\put(645,2756){\ellipse{450}{450}}
\put(645,1631){\ellipse{450}{450}}
\put(645,506){\ellipse{450}{450}}
\put(2895.000,4162.250){\arc{337.500}{2.4981}{6.9267}}
\blacken\path(3033.575,4184.641)(3030.000,4061.000)(3091.339,4168.413)(3033.575,4184.641)
\put(5145.000,4162.250){\arc{337.500}{2.4981}{6.9267}}
\blacken\path(5283.575,4184.641)(5280.000,4061.000)(5341.339,4168.413)(5283.575,4184.641)
\put(2895.000,3037.250){\arc{337.500}{2.4981}{6.9267}}
\blacken\path(3033.575,3059.641)(3030.000,2936.000)(3091.339,3043.413)(3033.575,3059.641)
\put(5145.000,3037.250){\arc{337.500}{2.4981}{6.9267}}
\blacken\path(5283.575,3059.641)(5280.000,2936.000)(5341.339,3043.413)(5283.575,3059.641)
\put(2895.000,1912.250){\arc{337.500}{2.4981}{6.9267}}
\blacken\path(3033.575,1934.641)(3030.000,1811.000)(3091.339,1918.413)(3033.575,1934.641)
\put(5145.000,1912.250){\arc{337.500}{2.4981}{6.9267}}
\blacken\path(5283.575,1934.641)(5280.000,1811.000)(5341.339,1918.413)(5283.575,1934.641)
\put(645.000,1912.250){\arc{337.500}{2.4981}{6.9267}}
\blacken\path(783.575,1934.641)(780.000,1811.000)(841.339,1918.413)(783.575,1934.641)
\put(2895.000,5766.500){\arc{4581.000}{1.1238}{2.0177}}
\blacken\path(3787.784,3624.521)(3885.000,3701.000)(3763.231,3679.267)(3787.784,3624.521)
\put(2895.000,3516.500){\arc{4581.000}{1.1238}{2.0177}}
\blacken\path(3787.784,1374.521)(3885.000,1451.000)(3763.231,1429.267)(3787.784,1374.521)
\put(1770.000,3037.250){\arc{337.500}{2.4981}{6.9267}}
\blacken\path(1908.575,3059.641)(1905.000,2936.000)(1966.339,3043.413)(1908.575,3059.641)
\put(2895.000,787.250){\arc{337.500}{2.4981}{6.9267}}
\blacken\path(3033.575,809.641)(3030.000,686.000)(3091.339,793.413)(3033.575,809.641)
\put(5145.000,787.250){\arc{337.500}{2.4981}{6.9267}}
\blacken\path(5283.575,809.641)(5280.000,686.000)(5341.339,793.413)(5283.575,809.641)
\put(4020.000,787.250){\arc{337.500}{2.4981}{6.9267}}
\blacken\path(4158.575,809.641)(4155.000,686.000)(4216.339,793.413)(4158.575,809.641)
\put(3435.000,3228.500){\arc{6440.662}{1.0915}{2.0501}}
\blacken\path(4825.906,290.710)(4920.000,371.000)(4799.193,344.436)(4825.906,290.710)
\put(645.000,3037.250){\arc{337.500}{2.4981}{6.9267}}
\blacken\path(783.575,3059.641)(780.000,2936.000)(841.339,3043.413)(783.575,3059.641)
\put(2895,3881){\ellipse{450}{450}}
\put(5145,3881){\ellipse{450}{450}}
\put(1770,3881){\ellipse{450}{450}}
\put(4020,3881){\ellipse{450}{450}}
\put(2895,2756){\ellipse{450}{450}}
\put(5145,2756){\ellipse{450}{450}}
\put(1770,2756){\ellipse{450}{450}}
\put(4020,2756){\ellipse{450}{450}}
\put(2895,1631){\ellipse{450}{450}}
\put(5145,1631){\ellipse{450}{450}}
\put(1770,1631){\ellipse{450}{450}}
\put(4020,1631){\ellipse{450}{450}}
\put(2895,506){\ellipse{450}{450}}
\put(5145,506){\ellipse{450}{450}}
\put(1770,506){\ellipse{450}{450}}
\put(4020,506){\ellipse{450}{450}}
\path(870,3881)(1545,3881)
\blacken\path(1425.000,3851.000)(1545.000,3881.000)(1425.000,3911.000)(1425.000,3851.000)
\path(4245,3881)(4920,3881)
\blacken\path(4800.000,3851.000)(4920.000,3881.000)(4800.000,3911.000)(4800.000,3851.000)
\path(4245,2756)(4920,2756)
\blacken\path(4800.000,2726.000)(4920.000,2756.000)(4800.000,2786.000)(4800.000,2726.000)
\path(4245,1631)(4920,1631)
\blacken\path(4800.000,1601.000)(4920.000,1631.000)(4800.000,1661.000)(4800.000,1601.000)
\path(870,506)(1545,506)
\blacken\path(1425.000,476.000)(1545.000,506.000)(1425.000,536.000)(1425.000,476.000)
\put(645,3814){\makebox(0,0)[b]{\smash{{\SetFigFont{6}{7.2}{\familydefault}{\mddefault}{\updefault}$1$}}}}
\put(645,2689){\makebox(0,0)[b]{\smash{{\SetFigFont{6}{7.2}{\familydefault}{\mddefault}{\updefault}$1$}}}}
\put(645,1564){\makebox(0,0)[b]{\smash{{\SetFigFont{6}{7.2}{\familydefault}{\mddefault}{\updefault}$1$}}}}
\put(645,439){\makebox(0,0)[b]{\smash{{\SetFigFont{6}{7.2}{\familydefault}{\mddefault}{\updefault}$1$}}}}
\put(1770,3814){\makebox(0,0)[b]{\smash{{\SetFigFont{6}{7.2}{\familydefault}{\mddefault}{\updefault}$2$}}}}
\put(2895,3814){\makebox(0,0)[b]{\smash{{\SetFigFont{6}{7.2}{\familydefault}{\mddefault}{\updefault}$3$}}}}
\put(4020,3814){\makebox(0,0)[b]{\smash{{\SetFigFont{6}{7.2}{\familydefault}{\mddefault}{\updefault}$4$}}}}
\put(5145,3814){\makebox(0,0)[b]{\smash{{\SetFigFont{6}{7.2}{\familydefault}{\mddefault}{\updefault}$5$}}}}
\put(15,3836){\makebox(0,0)[lb]{\smash{{\SetFigFont{6}{7.2}{\rmdefault}{\mddefault}{\updefault}(a)}}}}
\put(1770,2689){\makebox(0,0)[b]{\smash{{\SetFigFont{6}{7.2}{\familydefault}{\mddefault}{\updefault}$2$}}}}
\put(2895,2689){\makebox(0,0)[b]{\smash{{\SetFigFont{6}{7.2}{\familydefault}{\mddefault}{\updefault}$3$}}}}
\put(4020,2689){\makebox(0,0)[b]{\smash{{\SetFigFont{6}{7.2}{\familydefault}{\mddefault}{\updefault}$4$}}}}
\put(5145,2689){\makebox(0,0)[b]{\smash{{\SetFigFont{6}{7.2}{\familydefault}{\mddefault}{\updefault}$5$}}}}
\put(1770,1564){\makebox(0,0)[b]{\smash{{\SetFigFont{6}{7.2}{\familydefault}{\mddefault}{\updefault}$2$}}}}
\put(2895,1564){\makebox(0,0)[b]{\smash{{\SetFigFont{6}{7.2}{\familydefault}{\mddefault}{\updefault}$3$}}}}
\put(4020,1564){\makebox(0,0)[b]{\smash{{\SetFigFont{6}{7.2}{\familydefault}{\mddefault}{\updefault}$4$}}}}
\put(5145,1564){\makebox(0,0)[b]{\smash{{\SetFigFont{6}{7.2}{\familydefault}{\mddefault}{\updefault}$5$}}}}
\put(15,2711){\makebox(0,0)[lb]{\smash{{\SetFigFont{6}{7.2}{\rmdefault}{\mddefault}{\updefault}(b)}}}}
\put(15,1586){\makebox(0,0)[lb]{\smash{{\SetFigFont{6}{7.2}{\rmdefault}{\mddefault}{\updefault}(c)}}}}
\put(1770,439){\makebox(0,0)[b]{\smash{{\SetFigFont{6}{7.2}{\familydefault}{\mddefault}{\updefault}$2$}}}}
\put(2895,439){\makebox(0,0)[b]{\smash{{\SetFigFont{6}{7.2}{\familydefault}{\mddefault}{\updefault}$3$}}}}
\put(4020,439){\makebox(0,0)[b]{\smash{{\SetFigFont{6}{7.2}{\familydefault}{\mddefault}{\updefault}$4$}}}}
\put(5145,439){\makebox(0,0)[b]{\smash{{\SetFigFont{6}{7.2}{\familydefault}{\mddefault}{\updefault}$5$}}}}
\put(15,461){\makebox(0,0)[lb]{\smash{{\SetFigFont{6}{7.2}{\rmdefault}{\mddefault}{\updefault}(d)}}}}
\end{picture}
}
\end{center}
\caption[Transition graphs of $t$, $t'$, and $t_{Z''}$.]{Transition graphs of $t = [2,4,3,5,5]$, $t' = [1,4,3,5,5]$, and $t_{Z''} = [2,5,3,4,5]$.}
\label{fig:genJT}
\end{figure}
\end{example}

Now, by Propositions~\ref{prop:Jbound},~\ref{prop:cS}, and~\ref{prop:genJT}, we have

\begin{theorem}\label{thm:Jtrivial} 
Let $L \subseteq \Sig^*$ be a $\gJ$-trivial regular language with quotient complexity $n \ge 1$. Then its syntactic complexity $\sigma(L)$ satisfies $\sigma(L) \le g(n) = \lfloor e(n-1)! \rfloor$, and this bound is tight if $|\Sig| \ge 2^{n-1}$. 
\end{theorem}

\begin{remark}
It was shown by Saito~\cite{Sai98} that, if $S$ is a $\gJ$-trivial submonoid of~$\cF_Q$, then $\Orbit(S) = \{\Orbit(t) \mid t \in S\} \subseteq \rmPi_Q$ forms a $\join$-semilattice, called a \emph{$\gJ$-$\join$-semilattice}, such that $\Max(\Orbit(t) \join \Orbit(s)) = \Fix(t) \cap \Fix(s)$. Let $\cP_\join(\rmPi_Q)$ be the set of all $\gJ$-$\join$-semilattices that are subsets of $\rmPi_Q$. A maximal $\gJ$-trivial submonoid $S$ of $\cF_Q$ corresponds to a maximal element $P$ in $\cP_\join(\rmPi_Q)$, with respect to set inclusion, such that $S = \bigcup_{\pi \in P}\cE(\pi)$. $P \in \cP_\join(\rmPi_Q)$ is called \emph{full} if $\{\Max(\pi) \mid \pi \in P\} = \cP_n(Q)$, which is an maximal element in $\cP_\join(\rmPi_Q)$ with respect to set inclusion. The monoid $\cS_n$ then corresponds to a full $\gJ$-$\join$-semilattice, and hence it is maximal. Saito described all maximal $\gJ$-trivial submonoid of $\cF_Q$ and those corresponding to full $\gJ$-$\join$-semilattices. However, here we consider the $\gJ$-trivial submonoid of $\cF_Q$ with maximum cardinality. 
\end{remark}

\begin{remark}
The number $\lfloor e(n-1)! \rfloor$ also appears in the paper of Brzozowski and Liu~\cite{BrLiu12} as a lower bound and the conjectured upper bound for the syntactic complexity of definite languages. However, the semigroup $B_n$ with this cardinality in~\cite{BrLiu12} for definite languages is not isomorphic to~$\cS_n$, since $B_n$ is not $\gJ$-trivial. 
\end{remark}

\section{Quotient Complexity of the Reversal of $\gR$- and $\gJ$-Trivial Regular Languages}\label{sec:rev}

In this section we consider \emph{nondeterministic finite automata} (NFA's). An NFA $\cN$ is a quintuple $\cN = (Q, \Sig, \delta, I, F)$, where $Q$, $\Sig$, and $F$ are as in a DFA, $\delta : Q \times \Sig \to 2^Q$ is the nondeterministic transition function, and $I$ is the set of initial states. For any word $w \in \Sig^*$, the \emph{reverse} of $w$ is defined inductively as follows: $w^R = \eps$ if $w = \eps$, and $w^R = u^Ra$ if $w = au$ for some $a \in \Sig$ and $u \in \Sig^*$. The \emph{reverse} of any language $L$ is the language $L^R = \{w^R \mid w \in L\}$. For any finite automaton (DFA or NFA) $\cM$, we let $\cM^\rev$ denote the NFA obtained by reversing all the transitions of $\cM$ and exchanging the roles of initial and final states, and by $\cM^\deter$, the DFA obtained by applying the subset construction to~$\cM$  keeping only the reachable subsets. Then $L(\cM^\rev) = (L(\cM))^\rev$, and $L(\cM^\deter) = L(\cM)$. To simplify our proofs, we use an observation from~\cite{Brz62} that, for any NFA $\cN$ without empty states, if the automaton $\cN^\rev$ is deterministic, then the DFA $\cN^\deter$ is minimal. 

In~2004, Salomaa, Wood, and Yu~\cite{SWY04} showed that if a regular language $L$ has quotient complexity $n \ge 2$ and syntactic complexity $n^n$, then its reverse language $L^R$ has quotient complexity~$2^n$, which is maximal for regular languages. As shown in~\cite{BrYe11} and~\cite{BLY12}, for certain regular languages with maximal syntactic complexity in their subclasses, the reverse languages have maximal quotient complexity. This also holds for $\gR$- and $\gJ$-trivial regular languages. 

It was proved by Jir{\'a}skov{\'a} and Masopust~\cite{JiMa12} that, if $L$ is an $\gR$-trivial language with $n$ quotients, then $2^{n-1}$ is a tight upper bound on the quotient complexity of $L^R$, and this bound can be met if $L$ is a ternary language. Note that the syntactic semigroup of any $\gR$-trivial language is a subset of $\cF_Q$ for some set $Q$. Hence the upper bound $2^{n-1}$ on $\kappa(L^R)$ can also be reached if $L$ has $n$ quotients with maximal syntactic complexity $n!$.

For $\gJ$-trivial languages $L$, it was conjectured by Masopust\footnote{Personal communication} that, if $L$ has $n$ quotients, then the upper bound $2^{n-1}$ on the quotient complexity of $L^R$ can be reached using $n-1$ letters. This conjecture holds:

\begin{theorem}\label{thm:Jrev}
For $n \ge 2$, if $L$ is a regular $\gJ$-trivial language with quotient complexity $\kappa(L) = n$, then $\kappa(L^R) \le 2^{n-1}$. Moreover, this bound can be met by a language $L$ over an alphabet of size $n-1$.
\end{theorem}

\begin{proof}
Since any $\gJ$-trivial regular language is also $\gR$-trivial, the upper bound $2^{n-1}$ also holds for $\gJ$-trivial regular languages. 

To see that the bound is tight, consider the DFA $\cB_n = (Q, \Sig, \delta, 1, \{n\})$ such that $Q = \{1,\ldots,n\}$, $\Sig = \{a_1,\ldots,a_{n-1}\}$, where each $a_i$ defines the following transformation of $Q$: 
$$ja_i = j+1 \txt{for} 1 \le j \le i-1, ia_i = n, \txt{and} ja_i = j \txt{for} i+1 \le j \le n.$$
The DFA $\cB_n$ is minimal since, for each $i \in Q$, state $i$ can be reached by $a_{n-1}^{i-1}$, and the word $a_i$ is only accepted by state $i$. Let $L_n = L(\cB_n)$. Then $\kappa(L_n) = n$. 

Let $\cN_n = \cB_n^R$ be an NFA accepting $L_n^R$, which contains no unreachable states. The NFA $\cN_5$ is shown in Fig.~\ref{fig:JTrev}. Let $P$ be any subset of $Q$ containing~$n$. If $P = \{n\}$, then it is the initial set of states of~$\cN_n$. Otherwise, suppose $P = \{p_1,\ldots,p_k,n\}$, where $1 \le p_1 < \cdots < p_k < n$ and $1 \le k \le n-1$. Let $t = a_{p_1} \cdots a_{p_k}$ be a transformation of $Q$. Then, for any $j \in Q$, $jt = n$ if and only if $j \in~P$. Since $t \in \cT_{\cB_n}$, there exists a word $w \in \Sig^*$ such that $w$ performs the transformation~$t$, \ie, $t_w = t$. This means that, for any $p \in Q$, $\delta(p, w) = n$ if and only if $p \in P$. Hence we can reach the set $P$ of states of $\cN_n$ from the initial set of states by the word $w$. Since there are $2^{n-1}$ distinct subsets $P$ of $Q$ containing $n$, there are $2^{n-1}$ reachable states in $\cN_n^D$. 

\begin{figure}[hbt]
\begin{center}
\setlength{\unitlength}{0.00087489in}
\begingroup\makeatletter\ifx\SetFigFont\undefined%
\gdef\SetFigFont#1#2#3#4#5{%
  \reset@font\fontsize{#1}{#2pt}%
  \fontfamily{#3}\fontseries{#4}\fontshape{#5}%
  \selectfont}%
\fi\endgroup%
{\renewcommand{\dashlinestretch}{30}
\begin{picture}(4582,2391)(0,-10)
\put(1558,1587){\ellipse{382}{382}}
\put(2988,1592){\ellipse{382}{382}}
\put(4383,1592){\ellipse{382}{382}}
\put(2256,507){\ellipse{382}{382}}
\put(199,1587){\ellipse{382}{382}}
\put(199,1587){\ellipse{314}{314}}
\blacken\path(499.000,1617.000)(379.000,1587.000)(499.000,1557.000)(499.000,1617.000)
\path(379,1587)(1369,1587)
\blacken\path(1894.000,1617.000)(1774.000,1587.000)(1894.000,1557.000)(1894.000,1617.000)
\path(1774,1587)(2781,1587)
\blacken\path(3311.000,1617.000)(3191.000,1587.000)(3311.000,1557.000)(3311.000,1617.000)
\path(3191,1587)(4172,1587)
\blacken\path(455.018,1426.413)(334.000,1452.000)(428.740,1372.474)(455.018,1426.413)
\path(334,1452)(2089,597)
\blacken\path(1731.705,1325.111)(1639.000,1407.000)(1682.262,1291.119)(1731.705,1325.111)
\path(1639,1407)(2134,687)
\blacken\path(2855.738,1291.119)(2899.000,1407.000)(2806.295,1325.111)(2855.738,1291.119)
\path(2899,1407)(2404,687)
\blacken\path(4153.478,1373.415)(4249.000,1452.000)(4127.735,1427.612)(4153.478,1373.415)
\path(4249,1452)(2449,597)
\path(1774,462)(2075,462)
\blacken\path(1955.000,432.000)(2075.000,462.000)(1955.000,492.000)(1955.000,432.000)
\path(2899,1767)(2898,1770)(2895,1776)
	(2891,1786)(2885,1800)(2878,1817)
	(2871,1836)(2864,1856)(2858,1877)
	(2853,1899)(2849,1922)(2848,1945)
	(2849,1969)(2854,1992)(2861,2010)
	(2869,2026)(2877,2038)(2884,2047)
	(2891,2053)(2896,2058)(2902,2061)
	(2907,2063)(2912,2065)(2918,2067)
	(2925,2069)(2933,2072)(2944,2075)
	(2957,2078)(2972,2081)(2989,2082)
	(3006,2081)(3021,2078)(3034,2075)
	(3045,2072)(3053,2069)(3060,2067)
	(3066,2065)(3072,2063)(3076,2061)
	(3082,2058)(3087,2053)(3094,2047)
	(3101,2038)(3109,2026)(3117,2010)
	(3124,1992)(3129,1969)(3130,1945)
	(3129,1922)(3125,1899)(3120,1877)
	(3114,1856)(3107,1836)(3100,1817)
	(3093,1800)(3079,1767)
\blacken\path(3098.249,1889.186)(3079.000,1767.000)(3153.483,1865.753)(3098.249,1889.186)
\path(1459,1767)(1458,1770)(1455,1776)
	(1451,1786)(1445,1800)(1438,1817)
	(1431,1836)(1424,1856)(1418,1877)
	(1413,1899)(1409,1922)(1408,1945)
	(1409,1969)(1414,1992)(1421,2010)
	(1429,2026)(1437,2038)(1444,2047)
	(1451,2053)(1456,2058)(1462,2061)
	(1467,2063)(1472,2065)(1478,2067)
	(1485,2069)(1493,2072)(1504,2075)
	(1517,2078)(1532,2081)(1549,2082)
	(1566,2081)(1581,2078)(1594,2075)
	(1605,2072)(1613,2069)(1620,2067)
	(1626,2065)(1632,2063)(1636,2061)
	(1642,2058)(1647,2053)(1654,2047)
	(1661,2038)(1669,2026)(1677,2010)
	(1684,1992)(1689,1969)(1690,1945)
	(1689,1922)(1685,1899)(1680,1877)
	(1674,1856)(1667,1836)(1660,1817)
	(1653,1800)(1639,1767)
\blacken\path(1658.249,1889.186)(1639.000,1767.000)(1713.483,1865.753)(1658.249,1889.186)
\path(4294,1767)(4293,1770)(4290,1776)
	(4286,1786)(4280,1800)(4273,1817)
	(4266,1836)(4259,1856)(4253,1877)
	(4248,1899)(4244,1922)(4243,1945)
	(4244,1969)(4249,1992)(4256,2010)
	(4264,2026)(4272,2038)(4279,2047)
	(4286,2053)(4291,2058)(4297,2061)
	(4302,2063)(4307,2065)(4313,2067)
	(4320,2069)(4328,2072)(4339,2075)
	(4352,2078)(4367,2081)(4384,2082)
	(4401,2081)(4416,2078)(4429,2075)
	(4440,2072)(4448,2069)(4455,2067)
	(4461,2065)(4467,2063)(4471,2061)
	(4477,2058)(4482,2053)(4489,2047)
	(4496,2038)(4504,2026)(4512,2010)
	(4519,1992)(4524,1969)(4525,1945)
	(4524,1922)(4520,1899)(4515,1877)
	(4509,1856)(4502,1836)(4495,1817)
	(4488,1800)(4474,1767)
\blacken\path(4493.249,1889.186)(4474.000,1767.000)(4548.483,1865.753)(4493.249,1889.186)
\path(2179,327)(2178,324)(2175,318)
	(2171,308)(2165,294)(2158,277)
	(2151,258)(2144,238)(2138,217)
	(2133,195)(2129,172)(2128,149)
	(2129,125)(2134,102)(2141,84)
	(2149,68)(2157,56)(2164,47)
	(2171,41)(2176,36)(2182,33)
	(2187,31)(2192,29)(2198,27)
	(2205,25)(2213,22)(2224,19)
	(2237,16)(2252,13)(2269,12)
	(2286,13)(2301,16)(2314,19)
	(2325,22)(2333,25)(2340,27)
	(2346,29)(2352,31)(2356,33)
	(2362,36)(2367,41)(2374,47)
	(2381,56)(2389,68)(2397,84)
	(2404,102)(2409,125)(2410,149)
	(2409,172)(2405,195)(2400,217)
	(2394,238)(2387,258)(2380,277)
	(2373,294)(2359,327)
\blacken\path(2433.483,228.247)(2359.000,327.000)(2378.249,204.814)(2433.483,228.247)
\put(1572,1519){\makebox(0,0)[b]{\smash{{\SetFigFont{12}{14.4}{\familydefault}{\mddefault}{\updefault}2}}}}
\put(2989,1519){\makebox(0,0)[b]{\smash{{\SetFigFont{12}{14.4}{\familydefault}{\mddefault}{\updefault}3}}}}
\put(4384,1519){\makebox(0,0)[b]{\smash{{\SetFigFont{12}{14.4}{\familydefault}{\mddefault}{\updefault}4}}}}
\put(1009,912){\makebox(0,0)[b]{\smash{{\SetFigFont{12}{14.4}{\familydefault}{\mddefault}{\updefault}$a_1$}}}}
\put(1999,1137){\makebox(0,0)[b]{\smash{{\SetFigFont{12}{14.4}{\familydefault}{\mddefault}{\updefault}$a_2$}}}}
\put(2899,1137){\makebox(0,0)[b]{\smash{{\SetFigFont{12}{14.4}{\familydefault}{\mddefault}{\updefault}$a_3$}}}}
\put(3529,912){\makebox(0,0)[b]{\smash{{\SetFigFont{12}{14.4}{\familydefault}{\mddefault}{\updefault}$a_4$}}}}
\put(2256,439){\makebox(0,0)[b]{\smash{{\SetFigFont{12}{14.4}{\familydefault}{\mddefault}{\updefault}5}}}}
\put(199,1519){\makebox(0,0)[b]{\smash{{\SetFigFont{12}{14.4}{\familydefault}{\mddefault}{\updefault}1}}}}
\put(1999,102){\makebox(0,0)[b]{\smash{{\SetFigFont{12}{14.4}{\familydefault}{\mddefault}{\updefault}$\Sig$}}}}
\put(1549,2217){\makebox(0,0)[b]{\smash{{\SetFigFont{12}{14.4}{\familydefault}{\mddefault}{\updefault}$a_1$}}}}
\put(2989,2217){\makebox(0,0)[b]{\smash{{\SetFigFont{12}{14.4}{\familydefault}{\mddefault}{\updefault}$a_1,a_2$}}}}
\put(4429,2217){\makebox(0,0)[b]{\smash{{\SetFigFont{12}{14.4}{\familydefault}{\mddefault}{\updefault}$a_1,a_2,a_3$}}}}
\put(874,1677){\makebox(0,0)[b]{\smash{{\SetFigFont{12}{14.4}{\familydefault}{\mddefault}{\updefault}$a_2,a_3,a_4$}}}}
\put(2224,1677){\makebox(0,0)[b]{\smash{{\SetFigFont{12}{14.4}{\familydefault}{\mddefault}{\updefault}$a_3,a_4$}}}}
\put(3619,1677){\makebox(0,0)[b]{\smash{{\SetFigFont{12}{14.4}{\familydefault}{\mddefault}{\updefault}$a_4$}}}}
\end{picture}
}
\end{center}
\caption{NFA $\cN_5 = \cB_5^R$ for $n = 5$ accepting $L_5^R$.}
\label{fig:JTrev}
\end{figure}
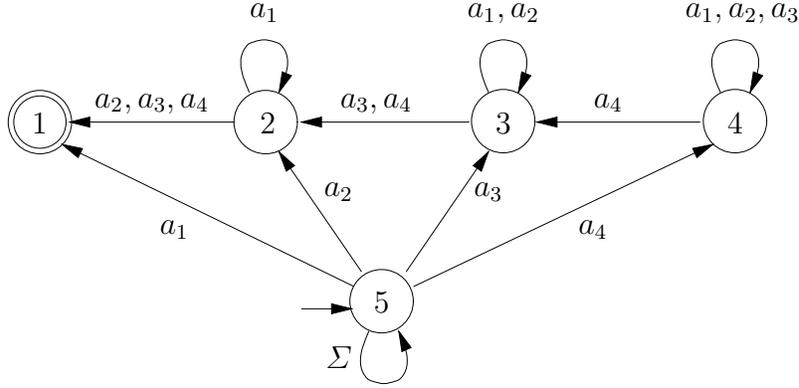

Note that there is no unreachable state in $\cB_n^R$. Then the DFA $\cN_n^D$ is minimal, and $\kappa(L_n^R) = 2^{n-1}$. This shows that the upper bound $2^{n-1}$ is tight for reversal of $\gJ$-trivial regular languages. \qed
\end{proof}

Consider again the above DFA $\cB_n$. The orbit of each transformation $a_i$ is $\{\{1,2,\ldots,i,n\},\{i+1\},\{i+2\},\ldots,\{n-1\}\}$; this is exactly the partition $\pi_Z$ for $Z = \{i+1,i+2,\ldots,n\}$. So $a_i \in \cS_n$ by definition. Then the transition semigroup of $\cB_n$ is a subsemigroup of $\cS_n$. It follows that, if a $\gJ$-trivial language $L$ has $n$ quotients and syntactic semigroup $\cS_n$, then its reverse $L^R$ has the maximal quotient complexity.

\section{Conclusion}\label{sec:con}

We proved that $n!$ and $\lfloor e (n-1)! \rfloor$ are the tight upper bounds on the syntactic complexities of $\gR$- and $\gJ$-trivial languages with $n$ quotients, respectively. For $n \ge 2$, the upper bound for $\gR$-trivial languages can be met using $1 + C^n_2$ letters, and the upper bound for $\gJ$-trivial languages, using $2^{n-1}$ letters. It remains open whether the upper bound for $\gJ$-trivial languages can be met with fewer than $2^{n-1}$ letters. The syntactic complexity of $\gL$-trivial languages is also open. 
We also observed that, if $\gR$- and $\gJ$-trivial languages have maximal syntactic complexities, their reverses have maximal quotient complexities. 
The proof of Theorem~\ref{thm:Jrev} can be extended to the following template for languages $L$ in some subclass $\cC$ of regular languages: Suppose $\cA = (Q, \Sig, \delta, q_1, F)$ is the minimal DFA of $L$. To prove $\kappa(L^R) = f(n)$, where $f(n)$ is an upper bound on $\kappa(L'^R)$ for $L' \in \cC$, one can show that there are at least $f(n)$ distinct subsets $P$ of $Q$ such that $\cA$ can perform a transformation $t$ of $Q$ with $it \in F$ if and only if $i \in P$. 

\newpage

\providecommand{\noopsort}[1]{}

\end{document}